\documentclass[journal,romanappendices]{IEEEtran}

\usepackage[caption=false,font=footnotesize]{subfig}
\usepackage{algorithm}
\usepackage{bm}
\usepackage{algpseudocode}
\usepackage{amsmath}
\usepackage{amsfonts}
\usepackage{amssymb}
\usepackage{amsthm}
\usepackage[hidelinks]{hyperref}
\usepackage{color}
\usepackage[detect-all,detect-weight=true]{siunitx}
\usepackage{booktabs}
\usepackage{pgfplots}
\usepackage{nicefrac}
\usepackage{cite}
\usepackage[showonlyrefs]{mathtools}

\DeclareSIUnit\bit{b}
\DeclareSIUnit\decibelI{dBi}

\DeclareFontFamily{U}{mathc}{}
\DeclareFontShape{U}{mathc}{m}{it}%
{<->s*[1.03] mathc10}{}

\DeclareMathAlphabet{\mathscr}{U}{mathc}{m}{it}

\usepgfplotslibrary{fillbetween, polar}
\usetikzlibrary{matrix, shapes.geometric, arrows, patterns} 
\pgfdeclarelayer{background}
\pgfsetlayers{background,main}

\tikzstyle{input} = [rectangle, rounded corners, minimum width=2cm, minimum height=1cm, text centered, draw=black, fill=orange!30]
\tikzstyle{output} = [rectangle, rounded corners, minimum width=2cm, minimum height=1cm, text centered, draw=black, fill=blue!30]
\tikzstyle{table} = [diamond, minimum width=2cm, minimum height=1cm, text centered, draw=black, fill=green!30]

\newcommand{\argmax}[1]{\underset{#1}{\operatorname{arg\,max\,}}}
\DeclareMathOperator{\diag}{diag}

\DeclareMathOperator{\E}{\operatorname{\mathbb{E}}}

\DeclareMathOperator{\cosTrunc}{\overline{\cos}}
\DeclareMathOperator{\sinTrunc}{\overline{\sin}}

\definecolor{lines-1}{RGB}{228,26,28}
\definecolor{lines-2}{RGB}{55,126,184}
\definecolor{lines-3}{RGB}{77,175,74}
\definecolor{lines-4}{RGB}{152,78,163}
\definecolor{lines-5}{RGB}{255,127,0}
\definecolor{lines-6}{RGB}{255,255,51}
\definecolor{lines-7}{RGB}{166,86,40}
\definecolor{lines-8}{RGB}{247,129,191}
\definecolor{lines-9}{RGB}{153,153,153}

\definecolor{verylightyellow}{RGB}{242,237,89}
\definecolor{verylightgray}{RGB}{200,200,200}
\definecolor{verylightred}{RGB}{255,191,191}
\definecolor{verylightgreen}{RGB}{117,234,123}

\pgfplotscreateplotcyclelist{myCycleList}{
	color=lines-1,thick,mark=o,mark size=3\\%
	color=lines-3,thick,mark=star,mark size=3\\%
	color=lines-4,thick,mark=square,mark size=3\\%
	color=lines-2,thick,mark=diamond,mark size=3\\%
	color=lines-5,thick,mark=triangle,mark size=3\\%
	color=lines-8,thick,mark=Mercedes star,mark size=3\\%
	color=lines-7,thick,mark=|,mark size=3\\%
	color=lines-9,thick,mark=pentagon,mark size=3\\%
}

\pgfplotsset{
	compat=1.14,
	width =\columnwidth, 
	height=.7\columnwidth,
	ylabel absolute, ylabel style={yshift=-0.30cm},
	xlabel absolute, xlabel style={yshift=0.2cm},
	label style={font=\footnotesize},
	tick label style={font=\footnotesize},
	legend style={font=\footnotesize},
	minor grid style={dotted},
}

\newtheorem{proposition}{Proposition}
\newtheorem{theorem}{Theorem}

\newcommand{\ed}{\color{black}} % for editorial notes ...

\makeatletter
\newcommand{\pushright}[1]{\ifmeasuring@#1\else\omit\hfill$\displaystyle#1$\fi\ignorespaces}
\newcommand{\pushleft}[1]{\ifmeasuring@#1\else\omit$\displaystyle#1$\hfill\fi\ignorespaces}
\newcommand{\subalign}[1]{%
	\vcenter{%
		\Let@ \restore@math@cr \default@tag
		\baselineskip\fontdimen10 \scriptfont\tw@
		\advance\baselineskip\fontdimen12 \scriptfont\tw@
		\lineskip\thr@@\fontdimen8 \scriptfont\thr@@
		\lineskiplimit\lineskip
		\ialign{\hfil$\m@th\scriptstyle##$&$\m@th\scriptstyle{}##$\crcr
			#1\crcr
		}%
	}
}
\newcommand{\vast}{\bBigg@{3.5}}
\makeatother

\allowdisplaybreaks[1]

\begin{document}
	
	\title{Cram\'er-Rao Bound  Analysis of Radars for Extended Vehicular Targets with Known and Unknown Shape}

	\author{\IEEEauthorblockN{
		Nil~Garcia, Alessio~Fascista, \IEEEmembership{Member, IEEE}, Angelo~Coluccia, \IEEEmembership{Senior Member, IEEE}, Henk~Wymeersch, \IEEEmembership{Senior Member, IEEE}, Canan~Aydogdu,  Rico~Mendrzik, \IEEEmembership{Member, IEEE}, Gonzalo~Seco-Granados, \IEEEmembership{Senior Member, IEEE}	% <-this % stops a space
% <-this % stops a space
\thanks{The work was supported in part by a seed project from Electrical Engineering Dept. at Chalmers and by the Swedish Research Council under grant 2018-03705, in part by Spanish project PID2020-118984GB-I00 and by the ICREA Academia programme.}
\thanks{N. Garcia, H. Wymeersch, and C. Aydogdu are with the Dept. of Electrical Engineering, Chalmers University of Technology, 412 96 Gothenburg, Sweden (e-mail: nil.garcia@gmail.com; henkw@chalmers.se; cananayy@gmail.com).}
\thanks{A. Fascista and A. Coluccia are with the Dept. of Innovation Engineering, Universit\`a del Salento, Via Monteroni, 73100 Lecce, Italy (e-mail: name.surname@unisalento.it).}
%\thanks{C. Aydogdu is with Ericsson AB, Gothenburg, Sweden (e-mail: canan.aydogdu@ericsson.com).}
\thanks{R. Mendrzik is with Ibeo Automotive Systems GmbH, 22143 Hamburg, Germany (e-mail: rico.mendrzik@ibeo-as.com).}
\thanks{G. Seco-Granados is with the Dept. of Telecommunications and Systems Engineering, Universitat Aut\`onoma de Barcelona, 08193 Barcelona, Spain (e-mail: gonzalo.seco@uab.cat).}
		%	\thanks{
		%		N.~Garcia and H.~Wymeersch are with the Department of Signals and Systems, and Q.~Fager is with the Department of Microtechnology and Nanoscience, Chalmers University of Technology, Gothenburg, Sweden.
		%		E.~G.~Larsson is with  the Division of Communication Systems, Department of Electrical Engineering (ISY), Link\"{o}ping University, Link\"{o}ping, Sweden.
		%		This research was supported in part, the VINNOVA COPPLAR project, funded under Strategic Vehicle Research and Innovation grant No.~2015-04849, 5Gcar, the Swedish Research Council (VR) and ELLIIT}
		}
	}
	
	\maketitle

	\begin{abstract}
	Due to their shorter operating range and large bandwidth, automotive radars can resolve many reflections from their targets of interest, mainly vehicles. This calls for the use of extended-target models in place of simpler and more widely-adopted point-like target models. However, despite some preliminary work, the fundamental connection between the radar’s accuracy as a function of the target vehicle state (range, orientation, shape) and radar properties remains largely unknown for extended targets. In this work, we first devise a mathematically tractable analytical model for a vehicle with arbitrary shape, modeled as an extended target parameterized by the center position, the orientation (heading) and the perimeter contour. We show that the derived expressions of the backscatter signal are tractable and correctly capture the effects of the  extended-vehicle shape. Analytical derivations of the exact and approximate hybrid Cram\'er-Rao bounds for the position, orientation and contour are provided, which reveal connections  with  the  case  of  point-like target and uncover the main dependencies with the received  energy, bandwidth, and array size. The theoretical investigation is performed on the two different cases of known and unknown vehicle shape. Insightful simulation results are finally presented to validate the theoretical findings, including an analysis of the diversity effect of  multiple radars sensing the extended target.
	
	\end{abstract}
	
	
	%	
	%	1. Automotive radar is increasingly important because of the explosion of autonomous driving and vehicular networks.
	%	3. Classical models for radar accuracy where derived for military radars that assumed point-targets, and therefore are unsuited.
	%	2. Most works focus on the algorithms, however, few works exist that shed light on the performance limits of extended-targets. Many models for extended-targets exist by their goal is to be as realistic as possible and therefore they are too nuanced to lead to anything more than numerical results.
	%	4. Vehicular networks involving radars are poorly modeled because of the lack of simple formulas for range and direction accuracy, and therefore, and for this reason the literature is scarce.
	%	5. This work walks the fine line between the simplicity of point-target models and the unbearable complexity of extended-target models, and provides a framework for analzying the performance under some mild assumptions.
	%	6. Interesting theoretical/numeric results are obtained including conditions under which the extended-target can be approximated as a point-target.

	%	\begin{enumerate}
	%		\item Mention my signal model offers a trade-off between a totally realistic signal and a tractable signal. Some approximations will be made for the sake of tractability.
	%		\item This study aims to fill the gap in fundamental analysis in extended targets.
	%	\end{enumerate}
	
		\bstctlcite{IEEEexample:BSTcontrol}

	\section{Introduction}
	
	Automotive radars are becoming the norm in modern cars thanks to their capability to estimate the speed and position of objects in the vicinity \cite{patole2017automotive,ARSP_2021,SPM_2019, SPM_2019b}. Today, most automotive radar systems operate  between \SIrange{76}{81}{\giga\hertz} which is part of the so-called millimeter-wave (mmWave) spectrum.  They are used in many applications of advanced driver assistance systems (ADAS) such as lane change assistance, automatic park control, cruise control, and are expected to become one of the leading technology for autonomous driving \cite{dickmann2016automotive,ITS_2020,Poor_ADAS}. Beyond the direct applications of radars on single vehicles, radars are also enablers for some forthcoming functionalities of vehicular networks, such as cooperative positioning \cite{frohle2018cooperative,soatti2018implicit,T_ITS} or platooning \cite{tsugawa2016review}.
	Thus, it is in the best interest of researchers/engineers to develop good models for analyzing radars' sensing accuracy in a vehicular environment.
	
	In the classical radar literature, the target is often assumed so distant that its reflection appears to come effectively from a single point in space, and as such it can be modeled by just four parameters: range, direction, radar cross-section and velocity. For instance, a typical military search radar detects aircraft at ranges in excess of \SI{400}{\kilo\meter} \cite{mark2010principles}, using adaptive techniques \cite{CFAR-FP,KNN,KNN_2,RangeSpread}.
	In agreement with well known principles of radars \cite[Ch.~11.3]{skolnik1981introduction}, the Cram\'{e}r-Rao bound (CRB) on the point-target model shows that the range and direction variance are inversely proportional to the square of the signal bandwidth and the array aperture, respectively \cite{brennan1961angular,miller1978modified}. The same dependencies have been observed in the context of non-radar based cooperative wireless localization \cite{Shen1}, also considering the effects of multipath propagation \cite{Shen2}.
	On the other hand, automotive radars with high range/angle accuracy operate much closer to their targets (e.g., vehicles) and can resolve many reflections around the objects \cite{dickmann2016automotive}. Such targets can no longer be described as single points in the space and are referred to as extended targets. 
	In the literature on automotive radars, there are studies that propose novel extended-target models \cite{granstrom2016extended, ICASSP2021}, advanced tracking algorithms \cite{hammarstrand2012extended}, and measure the variance bounds on some parameters of interest such as range and direction \cite{zhang2005dynamic}. However, the fundamental connection between the radar's accuracy as a function of the vehicle state (range, orientation, shape, velocity) and radar properties (emitted power, bandwidth, waveform) for extended targets remains largely unknown.
	%	The lack of proper models for the accuracy of automotive radars throttles the applicability of some studies on vehicular networks.

	A simple approach is to avoid modelling the target contour, and then infer the extended-target kinematic properties directly from the radar measurements. For instance, in \cite{Karl1} it has been shown that the main target parameters can still be estimated without modelling its shape, but only when the target is small enough. In \cite{buhren2006simulation}, the main reflection points of a vehicle (including the wheels) are modeled based on real world observations. Although this approach is  quite flexible, the achieved accuracy is often very limited and, moreover, it cannot capture important information related to the turning maneuvers of the target. By introducing a model for the shape of the extended target, it becomes possible to capture rigid rotations and in turn to infer more accurate position and kinematic information. To describe an extended target, different contour models exist in the radar literature. A common approach is to assume a specific basic geometric contour for the target, such as a rectangle \cite{knill2016direct, RectShape, RectShape2, Karl2}, a circle \cite{Circle1,Circle2}, or an ellipse \cite{ICASSP2021, Ellipse1, Ellipse2}. While such models allow for an elegant formulation and resolution of tracking problems, typical vehicular targets cannot be accurately represented by a simple geometric shape. When it comes to modelling targets with arbitrary shapes, the problem becomes significantly more complex and the approaches available in literature follow two main strategies. A first type of methods models the target contour through deterministic or stochastic curves parameterized by a set of chosen parameters. For instance, in \cite{Parametric1} deformations from the basic geometric shapes are considered to model the extended-target contours, while in \cite{GaussProc_Shape, GaussProc2_Shape, GP1, GP2} probabilistic contour models based on Gaussian processes are considered. Random hypersurfaces \cite{RHM_Shape} and B-spline curves \cite{BSpline_Shape, B-Spline1} have been also considered valid options for modelling extended targets with arbitrary shapes. On the other hand, the second type of approaches considers a combination of multiple ellipses to obtain a more accurate and detailed representation of the target contour \cite{Karl3, MultipleEllipses1}. The adoption of these extended-target models in vehicular tracking tasks allows to capture useful details such as the rounded corners of targets, but typically comes at the price of a significantly increased complexity and lack of interpretability in terms of the main radar parameters. 

In this work, we build on the approach used in \cite{staib1989parametrically} and derive a    model that is able to capture the backscatter effects generated by a waveform impinging on an extended target with arbitrary shape, while keeping the number of parameters used to describe the contour tractable. Specifically, we introduce a two-dimensional geometrical contour model whose components are parameterized by a truncated Fourier series with a small number of coefficients, which naturally incorporates the prior knowledge that vehicular targets are symmetric, allowing to correctly infer their main parameters even when the radar does not illuminate the whole contour. The derived expressions of the backscatter signals are then used to carry out a  Fisher information analysis aimed at investigating the theoretical accuracy achievable in the estimation of the parameters of interest (namely range, direction and orientation), using the mathematically tractable tool of the hybrid CRB (HCRB). In \cite{xu2009hybrid}, a HCRB is proposed for tracking ground-moving targets; however, the vehicle contour is assumed to be perfectly known a priori and is described by a simple rectangular shape. Furthermore, the considered signal model does not take into account the radar characteristics such as the signal bandwidth and the number of receive antennas. 
Other studies compute the posterior CRB with recursive observations \cite{tichavsky1998posterior,zhong2010comparison, PCRB_extend} but their formulas end up being too complex to extract any type of intuition.
	
		This work aims at partially closing the knowledge gap between the simplicity of point-like target models and the unbearable complexity of extended-target models. Specifically, it provides a novel framework for analyzing the performance of an automotive radar sensing the range, direction and orientation information. The main idea consists in modeling the vehicle as an extended target parameterized by a set of unknown parameters: the position of its center, the orientation (heading) and the perimeter contour (as shown in Fig.~\ref{fig:geometric_model}). For tractability, we consider a sufficiently short observation time over which the vehicle can be considered static.
		We address the more general scenario in which the vehicle contour can be arbitrary and unknown, and generates multiple reflections according to the specific portion that has been illuminated. 
	%We are interested in the case of automotive radars with precise angle/range resolution, thus, we  explicitly include the effect of the unknown vehicle shape into our derivations. 
	
The main contributions of this work are as follows.
\begin{itemize}
    \item A novel \emph{mathematically tractable  analytical  model}  for  an extended target  with arbitrary  shape is proposed,  parameterized   by   the   center position, the orientation (heading) and the perimeter contour. Remarkably, it is shown that the   model is sufficiently rich to  capture the backscatter effects of the extended-vehicle shape, despite the quite complex scenario at hand. 
\item A \emph{fundamental Fisher information analysis} based on the HCRB is conducted, which allows to uncover the main dependencies of range, direction, and orientation estimation upon received energy, effective bandwidth, array's effective aperture, and reflection coefficient. 
\item The impact of \emph{lack of knowledge of the target contour} onto the achievable radar localization performance is investigated, by deriving the HCRB for the two different cases of known and unknown contour. Based on such an analysis, interesting connections are drawn between extended and point-like targets. 
\item The diversity effect of having \emph{multiple radars sensing the extended target} is finally investigated, which turns out to improve significantly the localization accuracy, especially when the target contour is unknown and should be inferred from scratch.
\end{itemize}

	The rest of the paper is organized as follows. In Section \ref{sec:signal_model}, we derive the models for the vehicle contour and the received signal. From the models, the general HCRB is obtained  in Section~\ref{sec:HCRB}. In Section~\ref{sec:HCRB_asymp}, we provide approximate closed-form expressions for the HCRB in case of known and unknown target shape, together with a thorough analysis including a comparison with the CRB of a point-like target. In  Section~\ref{sec:simulations} we provide numerical results that show the correctness of the derived expressions. Conclusions are given in Section~\ref{sec::conclusions}.
	
	\begin{table}[ht]
\caption{Table of Main Symbols}
\small
%\centering
\begin{tabular}{c|l}
\hline\hline
\textbf{Symbol} & \textbf{Description} \\ [0.5ex] % inserts table %heading
\hline
$\mathcal{C}$ &Extended-target contour \\
$Q$ & Number of Fourier series coefficients \\
$a_q$&Contour coefficients along the $x$-axis \\
$b_q$ &Contour coefficients along the $y$-axis \\
$\mathbf{m}$ & Vector of $a_q$ coefficients  \\
$\mathbf{n}$ & Vector of $b_q$ coefficients  \\
$\boldsymbol{\sigma}$ & Vector of cosine harmonics \\
$\boldsymbol{\varsigma}$ & Vector of sine harmonics \\
$\mathbf{p}$ & Position of the center of target, eq. \eqref{eq:p}  \\
$\mathring{d}$ & Distance from center of target\\
$\mathring{\phi}$ & Direction of the center of target \\
$\varphi$ & Orientation (heading) of target \\
$u$ & Variable spanning the contour for $0 \leq u < 2\pi$\\
$\boldsymbol{\rho}(u)$ & Single point of contour in local coordinates, eq. \eqref{eq:perim}\\
$\mathbf{r}(u)$ & Single point of contour in global coordinates, eq. \eqref{eq:perimeter} \\
$d(u)$ & Distance from point $\mathbf{r}(u)$ on  target contour, eq. \eqref{eq::defdu} \\
$\phi(u)$ & Direction of point $\mathbf{r}(u)$ on  target contour in \eqref{eq::defphiu}\\
$\psi(u)$ & Angle between contour normal and LOS, eq. \eqref{eq:paremeters:theta}\\
$\mathbf{R}$ & Rotation matrix of the reference frame, eq. \eqref{eq:rotation}  \\
$s(t)$ & Radar transmitted waveform \\
$\mathrm{d}P(t)$ & Differential received power, eq. \eqref{eq:dP} \\
$B$ & Signal bandwidth \\
$B_\text{RMS}$ & Effective bandwidth \\
$f_c$ & Carrier frequency \\
$\lambda$ & Signal wavelength \\
$\Upsilon(\phi)$ & Radar antenna directivity \\
$\alpha$ & Target surface roughness \\
$G$ & Receiver gain \\
$P(t)$ & Aggregated power at the radar \\
$\ell_\text{T}$ & Total target perimeter \\
$K$ & Number of disjoint segments of target contour \\
$\ell_\text{R}$ & Length of each disjoint contour segment \\
$\mathcal{C}_k, d_k, \ldots$ & Parameters related to the $k$-th  disjoint segment\\
$h_k$ & Channel coefficient along the $k$-th path \\
$N$ & Number of antenna elements \\
$\mathbf{a}(\phi)$ & Array response vector \\
$\mathbf{e}(t)$ & Signal across antenna elements \\
$\mathbf{y}(t)$ & Received signal \\
$\mathbf{w}(t)$ & Additive white Gaussian noise \\
$T$ & Observation interval \\
$E$ & Received energy \\
$N_0$ & Noise power spectral density \\
$\mathbf{h}$ & Unknown random channel parameters in \eqref{eq::h} \\
$g$ & Deterministic nuisance channel parameter \\
$\boldsymbol{\gamma}$ & Vector of deterministic parameters of interest in \eqref{eq::gamma} \\
$\boldsymbol{\theta}$ & Vector containing all unknown parameters $\mathbf{h}, g, \boldsymbol{\gamma}$ \\
$\mathbf{C}$ & Hybrid Cram\'{e}r-Rao bound matrix, eq. \eqref{eq:CRB_Fisher} \\
$\mathbf{J}(\boldsymbol{\theta})$ & Fisher information matrix, eq. \eqref{eq:Fisher:full} \\
$\mathbf{J}(\boldsymbol{\gamma})$ & Effective Fisher information matrix, eq. \eqref{eq:Fisher}\\
$\mathbf{S}_\text{n}$ & Upper-left block of $\mathbf{J}(\boldsymbol{\gamma})$, eq. \eqref{eq:partition}\\
$\mathbf{S}_\text{p}$ & Lower-right block of $\mathbf{J}(\boldsymbol{\gamma})$, eq. \eqref{eq:partition} \\
$K_r$ & Number of radars \\
$\mathbf{J}_r(\boldsymbol{\gamma})$ & Fisher information matrix of $r$-th radar, eq. \eqref{eq:Fisher:diversity} \\
$\mathbf{p}_r, \mathring{d}_r, \ldots$ & Parameters related to the $r$-th radar \\
$\mathbf{T}, \mathbf{T}_{i,j}$ & Matrices of the approximated EFIM, eq. \eqref{eq:large_distance} \\
[1ex]
\hline
\end{tabular}
\label{table:symbols}
\end{table}
	
\subsubsection*{Notation}

Boldface lower-case and upper-case letters refer to vectors and matrices, respectively, while roman letters (both lower- and upper-case) are used for scalar variables.
$\mathring{x}$ and $\mathring{\mathbf{x}}$/$\mathring{\mathbf{X}}$ are scalar or vector/matrix quantities referred to the center of the target contour, relative to the corresponding variables $x$ and $\mathbf{x}/\mathbf{X}$ on each point along the contour.
$\mathbb{R}$ is the set of 
real numbers, and $\mathbb{R}^{n\times m}$ is the Euclidean space of $(n\times m)$-dimensional 
real matrices (or vectors if $m=1$). 
$\mathrm{j} = \sqrt{-1}$ is the imaginary unit. $\Re(\cdot)$ and $\Im(\cdot)$ denote the real and imaginary parts of the complex argument (with parentheses often omitted). $|z|$ and $z^*$ denote the modulus and the complex conjugate of the complex number $z$, respectively. 
$\| \cdot \|$ is the Euclidean norm of a vector and $\odot$ denotes the Hadamard (element-wise) product between two vectors.
 $(\cdot)^\mathsf{T}$, $(\cdot)^\mathsf{H}$, and $(\cdot)^{-1}$  denote the transpose, transpose conjugate (Hermitian), and inverse of a matrix, respectively. The Little-O notation \cite{balcazar1989nonuniform} $\mathbf{X} + o(x)$ applies entry-wise to the matrix $\mathbf{X}$.
The identity matrix is indicated by $\mathbf{I}$, and $\diag(d_1,\ldots,d_n)$ represents a diagonal matrix with elements given by the argument variables.
$\mathbf{A} \succcurlyeq \mathbf{B}$ indicates that the  matrix $\mathbf{A} -\mathbf{B}$ is positive semi-definite.  $\bigtriangleup_{\mathbf{x}_1}^{\mathbf{x}_2} f \triangleq \frac{\partial }{\partial\mathbf{x}_1} \frac{\partial }{\partial\mathbf{x}_2^\mathsf{T}} f$ with $\mathbf{x}_1$ and $\mathbf{x}_2$  two arbitrary  vectors.
$\mathbf{x}_\perp = \left(\begin{smallmatrix}0 & -1 \\ 1 & 0\end{smallmatrix}\right) \mathbf{x}$ for any $\mathbf{x}\in\mathbb{R}^{2\times 1}$.
$\mathbb{E}[\cdot ]$ is the statistical expectation operator and $x \sim \mathcal{CN}(\mu,\sigma^2)$ defines a circularly symmetric complex normal random variable $x$ with mean $\mu$ and variance $\sigma^2$.
The function $\arctan(y,x)$ denotes the four-quadrant inverse tangent.
A bar over a function denotes the rectifier operator defined as $\overline{f}=\max (f,0)$.  A dot over a scalar or vector variable, i.e., $\dot{x}/\dot{\mathbf{x}}$, refers  to the derivative (or gradient) with respect to the inherent scalar independent variable (time or angle).
We also define (see Sec.~\ref{sec::EFIM}) the \emph{star product} as a suitable inner product $\langle f_1,f_2\rangle_\star$ between two $L^2([0,2\pi])$  functions, with induced  norm $\|f\|_\star$ and orthogonal projection of $f_1$ over $f_2$ defined as $\operatorname{P}_{f_2}(f_1) \triangleq {\langle f_1,f_2\rangle_\star} {\langle f_2,f_2\rangle_\star^{-1}} f_2$, while the projection on the complement space is $\operatorname{P}_{f_2}^\perp(f_1) \triangleq f_1 -\operatorname{P}_{f_2}(f_1)$. 
	For convenience, the star product, norm, and projections are also overloaded for vector functions.
	Furthermore, in Sec.~\ref{sec::HCRB_unknown} the star product is extended over the space $L^2([0,2\pi])\times L^2([0,2\pi])$ as
	$\langle(f_1,f_2),(g_1,g_2)\rangle_\star \triangleq \langle f_1,g_1 \rangle_\star +\langle f_2,g_2 \rangle_\star$, and similarly overloaded to vector functions. In Table \ref{table:symbols}, we report a list of the main symbols used throughout the paper.

	\section{Signal Model} \label{sec:signal_model}

	\subsection{Extended-Target Contour Model}\label{sec::contourmodel}

	\begin{figure}
		%\centering
		\includegraphics[width=\columnwidth]{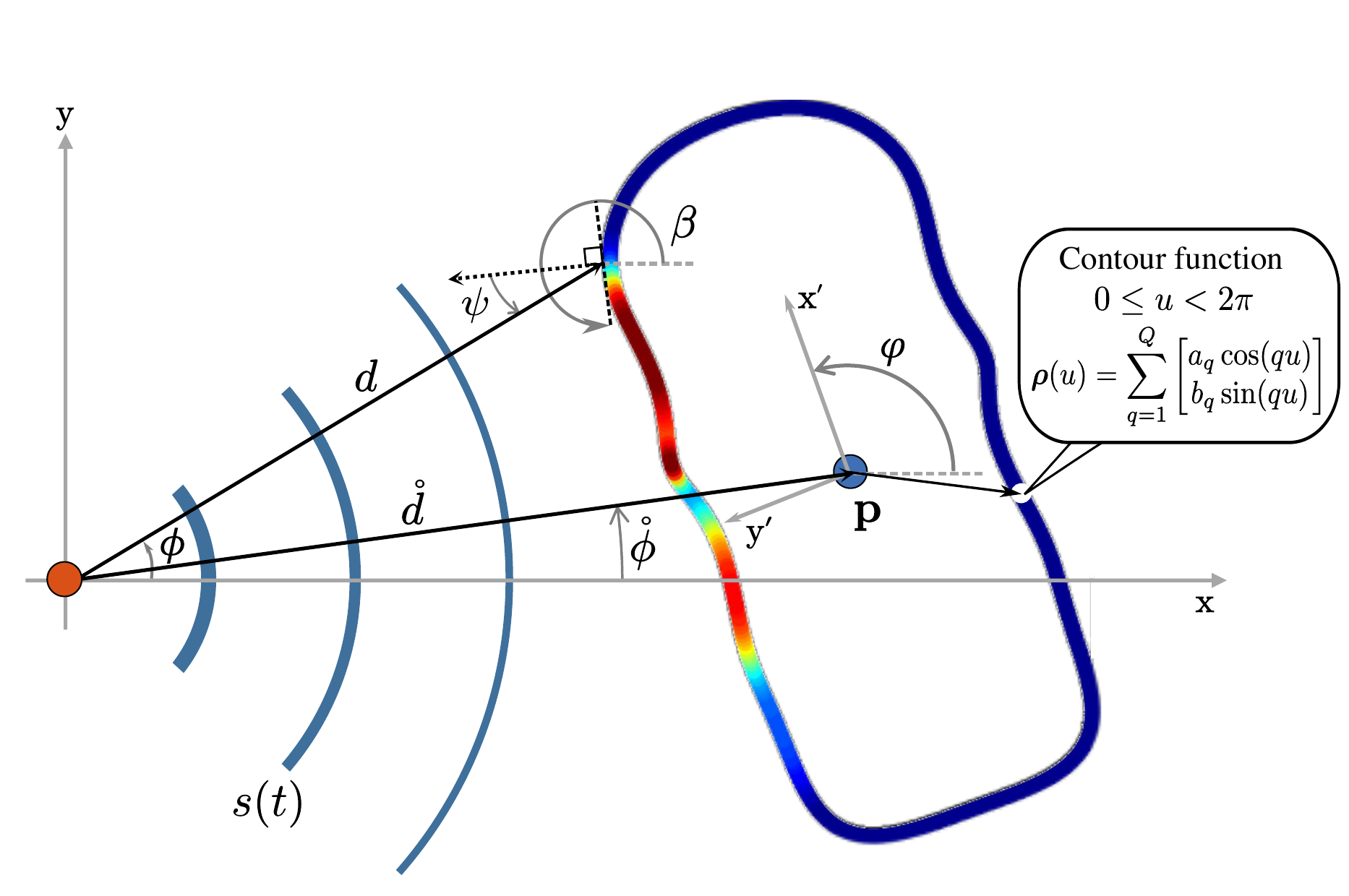}
		\caption{\ed Top view of an extended-target vehicle with orientation $\varphi = \SI{120}{\degree}$ and contour parameters $\mathbf{m} = [2.05,-0.02,0.17,0.05,-0.03,-0.01,-0.02,0.03,-0.01,-0.01]^\mathsf{T}$ and $\mathbf{n} = [1.12,0.005,0.24,-0.01,0.05,0.01,-0.01,-0.02,-0.02,0.014]^\mathsf{T}$. Colors indicate the average strength of the backscattered signal, from red (maximum strength) to blue (no energy reflected) based on $(\cosTrunc \psi)^{2(\alpha+1)}$ and $\alpha=2$. The radar is at the origin of the global coordinate system.}
		\label{fig:geometric_model}
	\end{figure}

	We preliminarily introduce an analytical model to describe a generic vehicle contour $\mathcal{C}$. As depicted in Fig.~\ref{fig:geometric_model}, we consider a birds-eye-view in a target local coordinate system where the vehicle position is assumed in the origin and the heading of the vehicle is aligned with the local $x$-axis. The two components of the contour are parameterized by a truncated Fourier series with $Q$ coefficients \cite{staib1989parametrically}. Specifically, the contour curve projected on the $x$-axis only includes cosine harmonics and the projection on the $y$-axis only includes sine harmonics because the vehicle's heading is in the direction of the positive $x$-axis and the contour is assumed symmetric with respect to the $x$-axis. Hence the perimeter is generated, for $0\leq u<2\pi$, by
	\begin{equation} \label{eq:perim}
	\boldsymbol{\rho}(u) = \sum_{q=1}^{Q} 
	\begin{bmatrix}
	a_q \cos(qu) \\ b_q \sin(qu)
	\end{bmatrix}
	%	\begin{pmatrix}
	%		a_q \cos(qu) & b_q \sin(qu)
	%	\end{pmatrix}^\mathrm{T},
	=\begin{bmatrix}
	\boldsymbol{\sigma}^\mathsf{T} \mathbf{m} \\ \boldsymbol{\varsigma}^\mathsf{T} \mathbf{n}
	\end{bmatrix},
	\end{equation}
	where $\boldsymbol{\sigma}\triangleq[\cos(u)\,\cdots\, \cos(Qu)]^\mathsf{T}$, $\boldsymbol{\varsigma}\triangleq[\sin(u)\,\cdots\, \sin(Qu)]^\mathsf{T}$, $\mathbf{m}\triangleq[a_1\,\cdots\, a_Q]^\mathsf{T}$ and $\mathbf{n}\triangleq[b_1\,\cdots\, b_Q]^\mathsf{T}$. The 0-th order harmonics have been omitted because the center of the vehicle is assumed to lie at the origin of the vehicle's local reference frame. 
	Without loss of generality, the first order harmonics are assumed to satisfy $a_1,b_1>0$ so that the contour cycles anti-clockwise. We highlight that the proposed model, through the use of a truncated Fourier series,  naturally incorporates a symmetric structure for a target, and can be thus used to describe arbitrary (symmetric) vehicle shapes with only a small set of coefficients according to \eqref{eq:perim}. The number of coefficients $Q$ determines the granularity of the target contour description. If one is interested in capturing even the finest details, then $Q$ should be large enough. Conversely, when fine-grained details are not of interest, or the available bandwidth is not large enough to observe them, $Q$ can be reasonably small.

 	By assuming that the radar is centered in the origin of the global coordinate system $Oxy$ as  in Fig.~\ref{fig:geometric_model}, a displacement of the vehicle to a given position $\mathbf{p}$ can be expressed as
	\begin{equation}
	    \mathbf{p}=[\mathring{d}\cos \mathring{\phi}\ \, \mathring{d}\sin \mathring{\phi}]^\mathsf{T} \label{eq:p}
	\end{equation} 
	where $\mathring{d}$ is the distance between the radar and the center of the vehicle while $\mathring{\phi}$ is the corresponding direction (ref. Fig.~ \ref{fig:geometric_model}). Considering a heading (orientation) $\varphi$ with respect to the $x$-axis, the contour model is modified as
	\begin{equation} \label{eq:perimeter}
	\mathbf{r}(u) = \mathbf{p} + \mathbf{R} \boldsymbol{\rho}(u),
	\end{equation}
	where $\mathbf{r}(u)=[r_x(u)\ r_y(u)]^\mathsf{T}$ as a function of $u$ describes the vehicle perimeter in the global coordinate system, and 
	\begin{equation}\label{eq:rotation}
	\mathbf{R} = \begin{bmatrix} \cos\varphi & -\sin\varphi \\ \sin\varphi & \cos\varphi \end{bmatrix}
	\end{equation} 
	is the rotation matrix. 
As a consequence, we finally define the vehicle contour as $\mathcal{C}=\left\{\mathbf{r}(u) : 0\leq u<2\pi \right\}$.

\subsection{Power Profile}

We consider a radar broadcasting a signal (e.g., a chirp sequence) through a single antenna, 
%\footnote{Notice that this is the worst case in terms of illumination of the target, which is performed as a whole at one shot. In case of a MIMO radar, it is conversely possible to illuminate small parts of the target sequentially, transmitting very directive signals in angular sweep fashion; this would be much more informative, thus significantly improving the bounds.Therefore, the provided framework can be  extended to cover the MIMO case, while the single antenna case addressed here represents in this respect an upper bound.},
whose returns are captured by a generic receive antenna array and digitally processed. In the following, in order to derive a model for the received signal based on the I/Q samples, we preliminarily devise analytical expressions for the received power.
	Let $\sqrt{E_\text{tx}}s(t)$ denote the transmitted waveform and, without loss of generality, let $\int|s(t)|^2\mathrm{d}t=1$. 
	Let $\Upsilon(\phi)$ be the radar antenna directivity towards the azimuth direction $\phi$, and let $\psi$ be the angle  between the normal vector to the vehicle surface and the line-of-sight (LOS) direction of the radar as shown in Fig.~\ref{fig:geometric_model}. Notice that $\Upsilon(\phi)$ is not specifically linked to a given type/model of automotive radar, but is used as a generic function representing any possible antenna directivity.
	The radar illuminates the entire target vehicle but, for the sake of the analysis, we focus on the azimuthal domain only\footnote{We address the problem in 2D, a common simplifying choice in the automotive radar literature that is tantamount to considering waves that propagate horizontally. This is realistic when the radar-target distance is large compared to the height of the antennas. Moreover, automotive radars are typically designed with a sufficiently wide fan beam so as to capture backscatter signals coming from targets with different elevations in their field of view.
Though the 3D case is outside the scope of the present contribution, nonetheless the proposed methodology can be extended to deal with such a case, provided that a 2D array is used in place of a ULA and the elevation angle is introduced in addition to the azimuth in the considered models.}. Moreover,  the transmitter antenna and receiver  array are assumed co-located and perfectly decoupled, so that the receiver does not suffer from self-interference due to full-duplex operation.
	%the height $h$ of the vehicle is assumed constant hence does not play any role in the derivation.
	Then, the differential received power $\mathrm{d}P(t)$ reflected by an infinitesimal part of the vehicle's contour $\mathcal{C}$ is given by \cite{kulmer2018impact}
	\begin{multline}
	\mathrm{d}P(t) \propto E_\text{tx} \left|s\left(t-\frac{2d(u)}{c}\right)\right|^2 \Upsilon^2(\phi(u))\\ 
	\times	\frac{\cosTrunc\psi(u)}{d^2(u)} \frac{\mathrm{d}u \cosTrunc\psi(u)}{d^2(u)} \frac{(1+\cos 2\psi(u))^{\alpha}}{2^\alpha (\alpha+1)}
		\label{eq:dP}
		%\mathrm{d}P(t) \propto \left(E_\mathrm{tx} \left|s\left(t-\frac{2d}{c}\right)\right|^2 g(\phi)\right) \\
	%	\left(\frac{\cosTrunc\theta}{d^2} \frac{h\mathrm{d}u \cosTrunc\theta}{d^2} \frac{(1+\cosTrunc 2\theta)^{\alpha}}{4\pi 2^\alpha (\alpha+1)}\right)
		%\left(\frac{\lambda^2}{4\pi} g(\phi) \right)
		%\label{eq:dP}
	\end{multline}
	%In the following, we assume $h = 1$ for simplicity. 
where we have introduced the modified cosine function 
%{\color{red} consider replacing both the truncated sin and cos with $\max(\sin/ \cos,0)$.}
% 	\begin{equation} \label{eq:cosTrun}
% 	\cosTrunc x \triangleq
% 	\begin{cases}
% 	\cos x & -\frac{\pi}{2} \leq x\; (\operatorname{mod} 2\pi) \leq \frac{\pi}{2} \\
% 	0 & \text{elsewhere}
% 	\end{cases}
% 	\end{equation}
	$ %\label{eq:cosTrun}
	\cosTrunc x \triangleq
	\max(\cos x,0)
$
% 	\begin{equation} \label{eq:cosTrun}
% 	\cosTrunc x \triangleq
% 	\operatorname{max}(\cos x,0)
% 	\end{equation}
to enforce the  assumption that no reflection occurs on the non-visible parts of the vehicle.
 It is worth noting that some of the parameters are a function of the reflection point $\mathbf{r}$, and consequently a function of $u$, i.e.
	\begin{subequations} \label{eq:paremeters}
		\begin{align}
		d(u) &= \|\mathbf{r}(u)\| \label{eq::defdu}\\
		\phi(u) &= \arctan(r_y(u),r_x(u))\label{eq::defphiu} \\
		\psi(u) &= 3\pi/2+\phi(u)-\beta(u), \label{eq:paremeters:theta}
		\end{align}
	\end{subequations}
	where $\beta(u) = \arctan(\dot{r}_y(u),\dot{r}_x(u))$
 and $\dot{\mathbf{r}}(u)=[\dot{r}_x(u)\ \dot{r}_y(u)]^\mathsf{T}$ with $\dot{r}_x(u) =  \dot{\boldsymbol{\sigma}}^\mathsf{T} \mathbf{m} \cos\varphi -  \dot{\boldsymbol{\varsigma}}^\mathsf{T}\mathbf{n} \sin\varphi$ and $\dot{r}_y(u) =  \dot{\boldsymbol{\sigma}}^\mathsf{T} \mathbf{m} \sin\varphi +  \dot{\boldsymbol{\varsigma}}^\mathsf{T}\mathbf{n}\cos\varphi$ (where the dot operator denotes in this case the derivative with respect to $u$). For brevity, the dependency on $u$ will be often omitted.
	
	The first part of \eqref{eq:dP}  accounts for the transmitted power at a given instant $t$ and the directivity of the antenna, with $d(u)$ given in \eqref{eq::defdu} denoting the distance between the radar and the point $\mathbf{r}(u)$ on the vehicle contour. The second part models the vehicle scattering as a function of the angles and the known surface roughness $\alpha$ \cite{kulmer2018impact}, and $ 2^\alpha (\alpha+1)$ is a normalizing factor that makes the total reflected power independent of $\alpha$, i.e., $[ 2^\alpha (\alpha+1)]^{-1}\int_{-\pi/2}^{\pi/2}(1+\cos2\psi)^{\alpha}\mathrm{d}\psi$ is a constant.  As an example, if the vehicle surface is completely reflective ($\alpha\rightarrow\infty$), the term $(1+\cos 2\psi)^{\alpha}/(2^\alpha (\alpha+1))$ becomes proportional to $\delta(\psi)$ meaning that all the power is reflected back, and if $\alpha=0$, the term implies isotropic scattering. 
Using the trigonometric identity $\cos 2\psi = 2\cos^2 \psi -1$ and $2\cos^2 \psi = 2\cosTrunc^2 \psi$ for $\psi \in [-\pi/2, \pi/2]$ (outside this interval the value of $1+\cos 2\psi$ is irrelevant being the two terms $\cosTrunc \psi$ in \eqref{eq:dP} zero), we can rewrite \eqref{eq:dP} in a more compact form as
	\begin{equation} \label{eq:rx_power}
	\mathrm{d}P(t) \propto \Upsilon^2(\phi) \frac{E_\text{tx}}{d^4} \frac{(\cosTrunc \psi)^{2(\alpha+1)}}{(\alpha+1)} \, \left|s\left(t-\frac{2d}{c}\right)\right|^2\mathrm{d}u
	\end{equation}
	which will be exploited and developed in the next sections.

	\subsection{Received Signal} \label{sec:received_signal}
	
	To develop a coherent yet tractable signal model, we make the assumptions that i) the target vehicle is in the field-of-view of the radar, and ii) the target vehicle is in the far-field of the radar's array because the near-field at mmWave frequencies is usually less than a meter.\footnote{The Fraunhofer distance\cite{std1979antennas} defines the beginning of the far field and its formula is $d_\text{F}=2(\text{array diameter})^2/\lambda$. For instance, the Fraunhofer distance for a ULA of 10 antennas with half-wavelength inter-antenna spacing operating at \SI{122}{\giga\hertz} \cite{jaeschke2014high} is \SI{12}{\centi\meter}.}
%	and (3) it exists a distance $\ell_\text{R}$, much larger than the signal wavelength $\lambda$, for which the reflected energy by any two points on the vehicle contour closer than $\ell_\text{R}$ is approximately the same.
	The aggregated power at the radar's antenna, due to the reflections along the vehicle contour, is assumed to be the superposition of independent paths originated from each infinitesimal element $\mathrm{d} \mathbf{r}$, i.e., $P(t) \triangleq \int_{\mathcal{C}}\mathrm{d}P$ can be written as a line integral 
	\begin{equation} \label{eq:aggregated_power}
	P(t) = 
	\tilde{G} \int_{\mathcal{C}} \Upsilon^2(\phi) \frac{E_\text{tx}}{d^4} (\cosTrunc \psi)^{2(\alpha+1)} \left|s\left(t-\frac{2d}{c}\right)\right|^2  \mathrm{d} \mathbf{r}
	\end{equation}
	where $\tilde{G}$ accounts for the receiver’s unknown gain and other constants not included in \eqref{eq:dP} or \eqref{eq:rx_power}.
	If the vehicle size is small compared to its range, then range and antenna element gain are approximately constant along the contour, i.e., $\forall u_1,u_2\in[0,2\pi]$, $d^{-4}(u_1)\approx d^{-4}(u_2) \approx \mathring{d}^{-4} $  and $\Upsilon^2(\phi(u_1))\approx \Upsilon^2(\phi(u_2))$, meaning that the specific shape of the radiation pattern does not have a significant impact. Conversely, we do not apply such an approximation to the delay of the baseband waveform in order to keep considering the dependency between  the distance $d$ and each point along the target contour. Indeed, the variations experienced by the delay of the baseband waveform as a function of $d$ along $\mathcal{C}$ bring the necessary information for estimating the distance.
	Accordingly, the aggregated power \eqref{eq:aggregated_power} from reflections along the contour simplifies to
	\begin{equation} \label{eq:aggregated_power:2}
	P(t) = \frac{G}{\mathring{d}^4} \int_{\mathcal{C}} (\cosTrunc \psi)^{2(\alpha+1)} \left|s\left(t-\frac{2d}{c}\right)\right|^2  \mathrm{d}\mathbf{r}
	\end{equation}
	where the unknown gain $G$ absorbed all the constants that act as scaling factors, hence are irrelevant to our analysis.
	%\footnote{Constants in the sense they do not change the signal as a function of time or across the antennas except for a scaling factor.}.
	
	Let $\ell_\text{T}$ be the total perimeter of the vehicle contour
	and for simplicity assume $\ell_\text{R}$ is a divisor of $\ell_\text{T}$. 
	Upon splitting the vehicle contour into $K\triangleq \ell_\text{T}/\ell_\text{R}$ disjoint continuous sections of length $\ell_\text{R}$, i.e., $\mathcal{C} = \bigsqcup_{k=1}^K \mathcal{C}_k$ with $\mathcal{C}_k \triangleq \left\{\mathbf{r}(u_k) : \tilde{u}_{k-1}\leq u_k<\tilde{u}_k \right\}$ and $0=\tilde{u}_0,\tilde{u}_1,\ldots,\tilde{u}_K=2\pi$ defining a partition of the interval $[0, 2\pi]$,
	the instantaneous power \eqref{eq:aggregated_power:2} can be approximated as a sum
	\begin{align}
	P(t) & = \frac{G}{\mathring{d}^4} \sum_{k=1}^{K}  \int_{\mathcal{C}_k}  (\cosTrunc \psi)^{2(\alpha+1)}  \left|s \!\left( t -  \frac{2d}{c}\right) \right|^2  \mathrm{d}\mathbf{r}_k \nonumber \\
	& \approx  \frac{G\ell_\text{R}}{\mathring{d}^4}
	\sum_{k=1}^{K}  P_k(t)  \label{eq:aggregated_power:sum}
	\end{align}
	where $\mathbf{r}_k \triangleq \mathbf{r}(u_k) \in\mathcal{C}_k$ for all $k$ and by treating the arguments of each integral over the contour sections $\mathcal{C}_k$ in the above expression as approximately constant with respect to the integration variable $u$, namely $(\cosTrunc \psi)^{2(\alpha + 1)} \approx (\cosTrunc \psi_k)^{2(\alpha + 1)}$ and $s(t-\frac{2d}{c}) \approx s(t - \frac{2d_k}{c})$, we have that
	\begin{equation}
	P_k(t) = (\cosTrunc \psi_k)^{2(\alpha+1)}
	\left|s\left(t-\frac{2d_k}{c}\right)\right|^2 \label{eq:Pk}
	\end{equation}
	where we abbreviated $d_k = d(u_k)$ and $\psi_k = \psi(u_k)$.
	Notice that in general $\int_\mathcal{C} (\cdot) \mathrm{d}\mathbf{r} = \int_0^{2\pi} (\cdot) \|\dot{\mathbf{r}}\| \mathrm{d}u$ hence likewise $\int_{\mathcal{C}_k} (\cdot) \mathrm{d}\mathbf{r}_k = \int_{\tilde{u}_{k-1}}^{\tilde{u}_k} (\cdot) \|\dot{\mathbf{r}}_k\| \mathrm{d}u_k$. 
	
	Starting from the instantaneous power, our goal is now to provide a model of the  signal received at the radar.
	If the electrical signal were modeled deterministically as a function of $u$, then one would simply take the square root of \eqref{eq:aggregated_power:sum}, but this would result in an intractable model. Exploiting the fact that signals originated from sufficiently far apart points (i.e., $\ell_\text{R}\gg\lambda$, where $\lambda$ is the signal wavelength) on the vehicle can be approximately treated as uncorrelated, a more convenient choice is to model the electrical signal $e(t)$ as a sum of independently random Rayleigh paths $e_k(t)$, i.e.,
	\begin{equation}
	e(t) =
	\frac{\sqrt{G \ell_\text{R}}}{\mathring{d}^2}
	\sum_{k=1}^{K} e_k(t)
	\end{equation}
	where
	$e_k(t)$ is obtained from \eqref{eq:Pk} as
	\begin{equation}
	e_k(t) = h_k (\sinTrunc \left(\phi_k-\beta_k\right))^{\alpha+1} s\! \left(t-\frac{2d_k}{c}\right)\end{equation}
	with
	$h_k\sim\mathcal{CN}(0,1)$ and $\E[|e_k(t)|^2]=P_k(t)$, upon defining 
% 	\begin{equation} 
% 	\sinTrunc x  
% 	\triangleq
% 	\begin{cases}
% 	\sin x & x\; (\operatorname{mod} 2\pi) \leq\pi \\
% 	0 & \text{elsewhere}
% 	\end{cases}
% 	\end{equation}
	$ 
	\sinTrunc x  
	\triangleq
	\max(\sin x ,0)
$
	so, by \eqref{eq:paremeters:theta},  $\cosTrunc(\psi_k) = \cosTrunc(3\pi/2+\phi_k-\beta_k) = \sinTrunc(\phi_k-\beta_k)$.

	All the derivations provided so far are valid for any arbitrary receive antenna at the radar side. The same expressions can be extended to the case of a radar receiving the backscattered signal with an antenna array. Specifically, if the target vehicle is in the far field and we assume that the radar is equipped with a uniform linear array (ULA) whose broadside points towards the positive $x$-axis, the array response to an incoming signal from azimuth $\phi$ is given by $\mathbf{a}(\phi) = [1\ \exp(-\mathrm{j}\pi\sin\phi)\ \cdots \ \exp(-\mathrm{j}\pi(N-1)\sin\phi)]^\mathsf{T}$, where $N$ is the number of antenna elements spaced a half-wavelength apart,  and we have used the usual narrowband assumption (i.e., $B \ll f_c$) which is reasonable due to the high carrier frequencies adopted by high-resolution automotive radars. Assuming that reflections originated within the same section result in approximately the same array response, then, the signal across the antenna elements of the phased array can be expressed as
	\begin{equation} \label{eq:received_signal}
	\mathbf{e}(t) =
	g\sqrt{\ell_\text{R}}
	\sum_{k=1}^{K} h_k
	\mathbf{a}(\phi_k) (\sinTrunc \left(\phi_k-\beta_k\right))^{\alpha+1} s\left(t-\frac{2d_k}{c}\right)
	\end{equation}
	where $g \triangleq \sqrt{G}/\mathring{d}^2$ and $\ell_\text{R}$ are left explicit for the convenience of the next section. 
	%\begin{equation} \label{eq:gain}
	%	G \triangleq \sqrt{(4\pi)^{-1}\tilde{G} P_\text{tx}}\mathring{d}^{-2}g(\mathring{\phi}).
	%\end{equation}
	%The constant $\sqrt{\ell_\text{R}}$ is left explicit (instead of aggregating it as part of $G$) for the convenience of the next section. {\color{red}It should be $g^2$}

	\section{Hybrid Cram\'{e}r-Rao Bound} \label{sec:HCRB}

	\subsection{ Fisher Information Matrix}
	
	The received signal can be expressed as $\mathbf{y}(t) = \mathbf{e}(t,\boldsymbol{\theta}) + \mathbf{w}(t)$, with $\mathbf{w}(t)$ denoting the additive white Gaussian noise having power spectral density  $N_0$ and $\mathbf{e}(t,\boldsymbol{\theta})$ from \eqref{eq:received_signal} having effective bandwidth $B_\text{RMS}$,
where the explicit dependence on the unknown vector $\boldsymbol{\theta}$ has been highlighted. 
	The unknown vector can be split into three parts as $\boldsymbol{\theta} \triangleq
	[
	\mathbf{h}^\mathsf{T} \, g \ \boldsymbol{\gamma}^\mathsf{T}
	]^\mathsf{T}$, where $\mathbf{h} \in \mathbb{R}^{2K \times 1}$ is a vector of nuisance random channel parameters,  $g$ is a deterministic nuisance channel parameter (ignoring its dependency on $\mathring{d}$), and $\boldsymbol{\gamma} \in \mathbb{R}^{(2Q + 3) \times 1}$ is the vector containing the deterministic parameters of interest, i.e.
	\begin{align}
	\mathbf{h} &\triangleq
	\begin{bmatrix}
	\Re h_1 & \Im h_1 & \cdots & \Re h_K & \Im h_K
	\end{bmatrix} ^\mathsf{T} \label{eq::h} \\
	\boldsymbol{\gamma} &\triangleq
	\begin{bmatrix}
	\mathring{d} & \mathring{\phi} & \varphi & a_1 & \cdots & a_Q & b_1 & \cdots & b_Q
	\end{bmatrix}^\mathsf{T}. \label{eq::gamma}
	\end{align}
	Ignoring terms that do not depend on $\boldsymbol{\theta}$, {\ed the log-likelihood function of the measurements is given (up to irrelevant additive constant terms) by
	\begin{equation} \label{eq:log-likelihood}
	\log p(\mathbf{y} |\boldsymbol{\theta})  =
	\frac{2}{N_0} \Re \int_{T} \mathbf{y}^\mathsf{H}(t) \mathbf{e}(t,\boldsymbol{\theta}) \,\mathrm{d}t
	-\frac{1}{N_0} \int_{T} \left\|\mathbf{e}(t,\boldsymbol{\theta}) \right\|^2 \mathrm{d}t
	\end{equation}}
	where $T$ is the observation interval.
	The HCRB on all parameters is obtained by inverting the Fisher information matrix (FIM) \cite{tichavsky1998posterior}
	\begin{align}
	\operatorname{Cov}\left(\boldsymbol{\theta}\right) &\succcurlyeq \mathbf{C} \label{eq:cov_SD}\\
	\mathbf{C} &\triangleq \mathbf{J}^{-1}\left(\boldsymbol{\theta}\right) \label{eq:CRB_Fisher}
	\end{align}
	 and the FIM is given by
	\begin{equation} \label{eq:Fisher:full}
	\mathbf{J}\left(\boldsymbol{\theta}\right) =
	-\E \left[ \bigtriangleup_{\boldsymbol{\theta}}^{\boldsymbol{\theta}} \log 
	p\left(\mathbf{y}, \mathbf{h} | g, \boldsymbol{\gamma}\right) \right]
	\end{equation}
	where $p(\mathbf{y}, \mathbf{h} | g, \boldsymbol{\gamma})$ is the joint a posteriori probability density function of the radar measurement. 
	
% 	\begin{lemma}
% 		The expression of the Fisher information matrix is
% 		\begin{align}
% 		\mathbf{J}\left(\boldsymbol{\theta}\right)\nonumber
% 		&=
% 		\begin{bmatrix}
% 		\mathbf{S}_\mathrm{n} & \mathbf{0} & \mathbf{0} \\
% 		\mathbf{0} & s_{11} & \mathbf{s}_{21}^\mathsf{T} \\
% 		\mathbf{0} & \mathbf{s}_{21} & \mathbf{S}_{22}
% 		\end{bmatrix}
% 		\end{align}
% 		where $\mathbf{S}_\mathrm{n} = -\E \left[ \bigtriangleup_{\mathbf{h}}^{\mathbf{h}} \log p(\mathbf{y} |\boldsymbol{\theta}) \right]
% 		-\E \left[ \bigtriangleup_{\mathbf{h}}^{\mathbf{h}}\log p(\mathbf{h})\right]$ $s_{11} = -\E \left[ \bigtriangleup_{g}^{g} \log p(\mathbf{y} |\boldsymbol{\theta}) \right]$, $\mathbf{s}_{21}  = -\E \left[ \bigtriangleup_{\boldsymbol{\gamma}}^{g} \log p(\mathbf{y} |\boldsymbol{\theta}) \right]$, and $\mathbf{S}_{22} = -\E \left[ \bigtriangleup_{\boldsymbol{\gamma}}^{\boldsymbol{\gamma}} \log p(\mathbf{y} |\boldsymbol{\theta}) \right]$.
% 	\end{lemma}
	
% 	\begin{proof}
% 		See Appendix~\ref{app:FIM}.
% 	\end{proof}
	
	By resorting to the Bayes theorem, we get
	\cite[Ch.~2.4.3]{van2004detection}
	\begin{align} \label{eq:Fisher_with_Bayes}
	\mathbf{J}\left(\boldsymbol{\theta}\right) &=
	-\E \left[ \bigtriangleup_{\boldsymbol{\theta}}^{\boldsymbol{\theta}} \log 
	p\left(\mathbf{y} | \mathbf{h}, g, \boldsymbol{\gamma}\right) \right]
	-\E \left[\bigtriangleup_{\boldsymbol{\theta}}^{\boldsymbol{\theta}} \log 
	p\left(\mathbf{h} | g, \boldsymbol{\gamma}\right) \right] \nonumber \\
	&= -\E \left[ \bigtriangleup_{\boldsymbol{\theta}}^{\boldsymbol{\theta}} \log 
	p\left(\mathbf{y} | \boldsymbol{\theta} \right) \right]
	-\E \left[\bigtriangleup_{\boldsymbol{\theta}}^{\boldsymbol{\theta}} \log 
	p\left(\mathbf{h} \right) \right].
	\end{align}
	For Gaussian observations, the expected second-order derivatives of the log-likelihood function are
	\begin{equation} \label{eq:log-likelihood_obs:derivatives}
	-\E \left[\bigtriangleup_{\mathbf{x}_1}^{\mathbf{x}_2} \log p(\mathbf{y} |\boldsymbol{\theta}) \right]
	%-\E \left[ \frac{\partial^2 \log p(\mathbf{y} |\boldsymbol{\theta})}{\partial \mathbf{x} \partial \mathbf{y}^\mathrm{T}} \right]
	= \frac{2}{N_0} \Re \int_{T} \E \left[ \frac{\partial\mathbf{e}^\mathsf{H}}{\partial \mathbf{x}_1}
	\frac{\partial\mathbf{e}}{\partial \mathbf{x}_2^\mathsf{T}} \right] \mathrm{d}t.
	\end{equation}
	We find that $\E \left[\bigtriangleup_{g}^{\mathbf{h}} \log p(\mathbf{y} |\boldsymbol{\theta})\right] = \mathbf{0}$ and $\E \left[\bigtriangleup_{\boldsymbol{\gamma}}^{\mathbf{h}} \log p(\mathbf{y} |\boldsymbol{\theta})\right] = \mathbf{0}$ because $\E \left[ h_k \right] = 0$; moreover, $\bigtriangleup_{g}^{\mathbf{h}} p(\mathbf{h}) = \mathbf{0}$ and $\bigtriangleup_{\boldsymbol{\gamma}}^{\mathbf{h}} p(\mathbf{h}) = \mathbf{0}$ since $\mathbf{h}$ is the only random vector of parameters.  
	Thus, the FIM can be partitioned as
	\begin{equation}
	\mathbf{J}\left(\boldsymbol{\theta}\right)	=
	\begin{bmatrix}
	\mathbf{S}_\text{n} & \mathbf{0} & \mathbf{0} \\
	\mathbf{0} & s_{11} & \mathbf{s}_{21}^\mathsf{T} \\
	\mathbf{0} & \mathbf{s}_{21} & \mathbf{S}_{22}
	\end{bmatrix}\label{eq:partition}
	\end{equation}
	where the upper-left block $\mathbf{S}_\text{n}$ depends solely on the  nuisance parameters, while the lower-right block of size $(2Q +4)\times(2Q+4)$, which refers to all the nonzero terms except for the $\mathbf{S}_\text{n}$ block, is related to the parameters of interest $\boldsymbol{\gamma}$. In the following, we denote such a lower-right block as $\mathbf{S}_\text{p}$ and focus on the part of $\mathbf{C}$ related to $\boldsymbol{\gamma}$.
%	where $s_{11} \triangleq -\E \bigtriangleup_{g}^{g} \tilde{p}$, $\mathbf{s}_{21} \triangleq -\E \bigtriangleup_{\bar{\boldsymbol{\theta}}}^{g} \tilde{p}$ and $\mathbf{S}_{22} \triangleq -\E \bigtriangleup_{\bar{\boldsymbol{\theta}}}^{\bar{\boldsymbol{\theta}}} \tilde{p}$.
	 We start by overloading the notation in \eqref{eq:CRB_Fisher} and refer to such a submatrix as $\mathbf{C}(\boldsymbol{\gamma})$. By the block matrix inversion formula:
	\begin{align}
	\mathbf{C}(\boldsymbol{\gamma}) &\triangleq \mathbf{J}^{-1}(\boldsymbol{\gamma}) \\
	\mathbf{J}(\boldsymbol{\gamma}) &= \mathbf{S}_{22} -\frac{1}{s_{11}} \mathbf{s}_{21} \mathbf{s}_{21}^\mathsf{T} \label{eq:Fisher:submatrix}
	\end{align}
	where $\mathbf{S}_\text{22}$ is a matrix of size $(2Q +3)\times(2Q+3)$, and $\mathbf{J}(\boldsymbol{\gamma})$ is the effective Fisher information matrix (EFIM), whose expression arises naturally in the process of inverting the block matrix as the Schur's complement of $s_{11}$ over $\mathbf{S}_\text{p}$.
	
	As a final remark, it is worth noting that the use of the HCRB tool together with a prior information on channel parameters (i.e., $h_k, k=1,\ldots,K$ are i.i.d. complex normal variables with zero mean and unit variance) leads to the more convenient partitioning of the FIM matrix in \eqref{eq:partition}, where the cross-correlation (off-diagonal) terms are zero and $\mathbf{J}(\boldsymbol{\gamma})$ does not depend on the specific realization of the random variables $h_k$.\footnote{The division of the whole target contour into $K$ disjoint segments is actually a convenient mathematical expedient that is used, together with the prior information on the channel parameters $h_k$ and the definition of HCRB, to obtain the block-diagonal structure of the FIM in \eqref{eq:partition}. As shown in Appendix I-A, the advantage brought by this equivalent representation is the possibility of being reverted during the derivations, allowing us to retrieve the initial representation in terms of the entire target contour $\mathcal{C}$.} {\ed Moreover, the peculiar structure of the FIM provided in \eqref{eq:partition} 
	remains valid for any choice of the statistical distribution of the random channel parameters (including the case of arbitrarily correlated parameters), with the only condition that $\E \left[ h_k \right] = 0 \ \forall k$.}
	%Conversely, using the general CRB definition in which all the parameters are treated as deterministic unknown would lead to a rather different FIM matrix, whose cross-correlation terms are generally different from zero. As a consequence, $\mathbf{J}(\boldsymbol{\gamma})$ would depend on the specific values assumed by the channel parameters $h_k$, so making the result less general.
	
	\subsection{ Effective Fisher Information Matrix}\label{sec::EFIM}
	
	In the following, we derive a closed-form expression of $\mathbf{J}(\bm{\gamma})$.
	%	\subsection{ Fisher Information Submatrix} \label{sec:FIS}
 Recalling that $\int_\mathcal{C} (\cdot) \mathrm{d}\mathbf{r} = \int_0^{2\pi} (\cdot) \|\dot{\mathbf{r}}\| \mathrm{d}u$,	we define the \emph{star product} as the inner product
	\begin{equation} \label{eq:star_product}
	\langle f_1,f_2\rangle_\star \triangleq \int_{0}^{2\pi}f_1(u)f_2(u) \left\|\dot{\mathbf{r}}(u)\right\| \,\mathrm{d}u 
	%		=\int_{0}^{2\pi}f(u)g(u) \|\dot{\boldsymbol{\rho}}(u)\|\mathrm{d}u
	\end{equation}
	over the  $L^2([0,2\pi])$ space of real square-integrable functions on the interval $[0,2\pi]$. The star product induces the norm $\|f\|_\star \triangleq  {\langle f,f \rangle_\star^{\nicefrac{1}{2}}} = (\int_{0}^{2\pi}f^2\|\dot{\mathbf{r}}(u)\|\mathrm{d}u)^{\nicefrac{1}{2}}$, and subsequently the orthogonal projection of $f_1$ over $f_2$ is $\operatorname{P}_{f_2}(f_1) \triangleq {\langle f_1,f_2 \rangle_\star} {\langle f_2,f_2 \rangle_\star^{-1}} f_2$, while the projection on the complement space is $\operatorname{P}_{f_2}^\perp(f_1) \triangleq f_1 -\operatorname{P}_{f_2}(f_1)$. 
	For convenience, the star product is overloaded such that for two vector functions $\mathbf{f}_1(u)$ and $\mathbf{f}_2(u)$, ${\langle\mathbf{f}_1,\mathbf{f}_2 \rangle_\star} \triangleq \int_{0}^{2\pi}\mathbf{f}_1(u)\mathbf{f}_2^\mathsf{T}(u) \|\dot{\boldsymbol{r}}(u)\|\mathrm{d}u$, where the $(m,n)$ component of the resulting matrix is the inner star product between the $m$-th entry of $\mathbf{f}_1$ and the $n$-th entry of $\mathbf{f}_2$. The projection operators are also overloaded: $\operatorname{P}_{\mathbf{f}_2}(\mathbf{f}_1) \triangleq {\langle\mathbf{f}_1,\mathbf{f}_2 \rangle_\star} {\langle\mathbf{f}_2,\mathbf{f}_2 \rangle_\star^{-1}} \mathbf{f}_2$ and $\operatorname{P}_{\mathbf{f}_2}^\perp(\mathbf{f}_1) \triangleq \mathbf{f}_1 -\operatorname{P}_{\mathbf{f}_2}(\mathbf{f}_1)$.
Moreover, we define%
\begin{subequations} \label{eq:w&v}
		\begin{align}\label{eq::defw}
		w &\triangleq (\sinTrunc (\phi-\beta))^{\alpha+1} \\
		v &\triangleq (\sinTrunc (\phi-\beta))^{\alpha} \cos(\phi-\beta).
		\end{align}
	\end{subequations}
The following Theorem provides a general EFIM expression.
	
	\begin{theorem}\label{prop1}
		A closed-form formula of the EFIM \eqref{eq:Fisher:submatrix} is
		\begin{multline} \label{eq:Fisher}
		%\mathbf{J}(\boldsymbol{\theta}_\text{i})  = \frac{2}{\|w\|_\star^2} \frac{2 E}{N_0} \Big[	w^2 L <\boldsymbol{\mu},\boldsymb l{\mu}>_\star +w^2 \mathrm{k} \cos^2(\phi) 
		\!\!\!\!\!\!\mathbf{J}(\boldsymbol{\gamma})  = \frac{2}{\|w\|_\star^2} \frac{2 E}{N_0} 
		\Big[
		 L
		\langle w\boldsymbol{\mu},w\boldsymbol{\mu} \rangle_\star
		+ M\langle w\cos\phi\boldsymbol{\eta},w\cos\phi\boldsymbol{\eta} \rangle_\star\\
		+ (\alpha+1)^2
		\langle \operatorname{P}_w^\perp\left(v\,\boldsymbol{\xi}\right), \operatorname{P}_w^\perp\left(v\,\boldsymbol{\xi}\right) \rangle_\star  \Big]
		\end{multline}
		%		$\boldsymbol{\mu} = \partial d/\partial\bar{\boldsymbol{\theta}}$, $\boldsymbol{\eta} = \partial \phi/\partial\bar{\boldsymbol{\theta}}$ and $\boldsymbol{\xi} = \boldsymbol{\eta} -\partial \beta /\partial\bar{\boldsymbol{\theta}}$.
		where the received energy 
		\begin{equation} \label{eq:energy}
		E \triangleq \int_{T}\|\mathbf{e}(t)\|^2\mathrm{d}t \approx g^2N \|w\|_\star^2	
		\end{equation} 
		calculated from \eqref{eq:received_signal} has been recognized, 
		the constants
		\begin{align}\label{eq::defr}
		L &\triangleq (4\pi B_\text{RMS}/c)^2 \\ \label{eq::defk}
		M &\triangleq \pi^2 (N^2-1)/12,
		\end{align}
		are related to the signal bandwidth and ULA length, respectively, 
		and the expressions of the vectors $\boldsymbol{\mu}$, $\boldsymbol{\eta}$ and $\boldsymbol{\xi}$ are found in  \eqref{eq::mu}, \eqref{eq::eta} and \eqref{eq::xi}, respectively. %Appendix~\ref{app:derivatives_intermediate_param}.
	\end{theorem}
	
	\begin{proof}
		See Appendix~\ref{app:FIS}.
	\end{proof}
	
	Theorem 1 provides a closed-form expression of the EFIM \eqref{eq:Fisher}, composed of three terms which are functions of $u$ and are integrated over the vehicle contour through the star norm operator. In particular, the first term $ L {\langle w\boldsymbol{\mu},w\boldsymbol{\mu} \rangle_\star}$ conveys information on the vehicle state from the terms related to the distance between the radar and the target in the observed signal $\mathbf{y}(t)$, and increases proportionally with the square of the signal bandwidth, being $L\propto B_\text{RMS}^2$. The second term $M  {\langle  w\cos\phi\boldsymbol{\eta},w\cos\phi\boldsymbol{\eta} \rangle_\star}$ represents 
	%directional information and increases proportionally to the square of the array aperture as seen from the target vehicle, $\mathrm{k}\cos^2(\phi) \propto (N^2-1)\cos^2(\phi)$ where $(N-1)|\cos(\phi)|\lambda/2$ is the array physical length projected to the target vehicle.
	 information related to the directions from which the backscatter signals, generated by reflections on the target contour, are received by the radar, and increases proportionally to the square of the array aperture, being $M \propto (N^2-1)$ where $(N-1)\lambda/2$ is the array physical length. The third and last term $(\alpha+1)^2 {\langle \operatorname{P}_w^\perp(v\boldsymbol{\xi}), \operatorname{P}_w^\perp(v\boldsymbol{\xi}) \rangle_\star}$ contributes to the information on the vehicle state via the signal-strength information associated to the energy reflected back from the target in the observation $\mathbf{y}(t)$, and depends explicitly on $\alpha$, the reflectivity of the surface. The latter parameter also impacts on the constant multiplying the  EFIM in \eqref{eq:Fisher} through the squared star norm $\|w\|_\star^2$ at the denominator.
	
	\vspace{-0.15cm}
	
	\subsection{CRB for a Point-like Target} A point-like target is a theoretical approximation of a target in which the received signal is modeled as coming from a point in space with zero extent. They are considered in the literature as valid approximations for very distant targets \cite{skolnik1981introduction}. Mathematically, keeping the notation consistent with \eqref{eq:received_signal}, the noise-free baseband signal model is
	\begin{equation} \label{eq:received_signal:point_target}
	\mathbf{e}(t) =
	g_\text{p.t.}\, \mathbf{a}(\mathring{\phi}) s\left(t-2\mathring{d}/c\right).
	\end{equation}
	The gain $g_\text{p.t.}$ models the radar cross-section of the target at the moment of the measurement and it is unpredictable in general. Parameters $\mathring{\phi}$ and $\mathring{d}$ are the direction and range, respectively, to the point-like target. The CRB for them is given by
	\begin{equation} \label{eq:Fisher:point_target}
	\mathbf{C}_\text{p.t.}(\mathring{d},\mathring{\phi}) \triangleq \mathbf{J}^{-1}_\text{p.t.}(\mathring{d},\mathring{\phi}) = \left(2\frac{E}{N_0}\right)^{-1} 
	\begin{pmatrix}
	L^{-1} & 0 \\ 0 & Z^{-1}
	\end{pmatrix}
	\end{equation}
	where $\mathbf{J}_\text{p.t.}(\mathring{d},\mathring{\phi})$ is the EFIM of a point-like target and
	\begin{equation} \label{eq:a}
	Z \triangleq M\cos^2(\mathring{\phi})
	\end{equation}
	where $M$ is given in \eqref{eq::defk}.
	The derivation of \eqref{eq:Fisher:point_target} is a much simpler version of that leading to \eqref{eq:Fisher} for the case of an extended target and is omitted for conciseness. It is very similar to that of the extended target but assuming only a single reflection point which greatly reduces the derivation.
	
	\subsection{Bounds for Multiple Radars}
	A common problem when dealing with extended targets is that the positioning accuracy is degraded since, from the radar's perspective, only a part of the target vehicle is visible. This problem is exacerbated when the contour is unknown because the radar has to infer it from scratch, also including the non-visible parts. Theoretically, it is possible to estimate the contour because targets are usually symmetric and the contour model is limited to $2Q$ coefficients. 
	To further investigate the accuracy in sensing the extended-target parameters, we extend our previous derivations to the case in which multiple radars are collaboratively available around the target. Since the direction and range of the target vehicle are relative to each radar, we derive the EFIM on the location of the vehicle in Cartesian coordinates. 
	Let $\mathbf{p}_r=[p_{x,r}\ p_{y,r}]^\mathsf{T}$ and $\kappa_r$ be the location and orientation, respectively, of radar $r$.
	Then, radar $r$ and the parameters of the target vehicle are related through $\mathbf{p} =  \mathbf{p}_r +\mathring{d}_r [\cos(\mathring{\phi}_r+\kappa_r)\ \sin(\mathring{\phi}_r+\kappa_r)]^\mathsf{T}$. By the chain rule and taking into account that the noise terms at the different radars are statistically independent, we obtain
	\begin{align}
	\mathbf{J}(p_x,p_y,&\varphi,\mathbf{m}^\mathsf{T},\mathbf{n}^\mathsf{T}) = \sum_{r=1}^{K_r}
	\mathbf{M}_r
	\mathbf{J}_r(\boldsymbol{\gamma}) 
	\mathbf{M}_r^\mathsf{T}  \label{eq:Fisher:diversity} \\
	\mathbf{M}_r &=
	\begin{bmatrix}
	\frac{\partial\mathring{d}_r}{\partial\mathbf{p}} & \frac{\partial\mathring{\phi}_r}{\partial\mathbf{p}} & \mathbf{0} \\
	\mathbf{0} & \mathbf{0} & \mathbf{I}
	\end{bmatrix} \nonumber \\
	&=
	\begin{bmatrix}
	\mathring{d}_r^{-1} \left(p_x-p_{x,r}\right) &
	\mathring{d}_r^{-2} \left(p_{y,r}-p_y\right) & \mathbf{0} \\
	\mathring{d}_r^{-1} \left(p_y-p_{y,r}\right) &
	\mathring{d}_r^{-2} \left(p_x-p_{x,r}\right) & \mathbf{0} \\
	 \mathbf{0} & \mathbf{0} & \mathbf{I}
	\end{bmatrix},
	\end{align}
	where $\mathbf{J}_r(\boldsymbol{\gamma})$ is the EFIM in \eqref{eq:Fisher} for radar $r$. Denoting by $\mathbf{C}(p_x,p_y,\varphi,\mathbf{m}^\mathsf{T},\mathbf{n}^\mathsf{T})$ the inverse of \eqref{eq:Fisher:diversity} and $C(p_x)$, $C(p_y)$ the scalar elements corresponding to the position coordinates (by the usual overloaded notation), the position error bound is defined as $\text{PEB} = [C(p_x) +C(p_y)]^{1/2}$; it represents a lower bound for the positioning accuracy of any unbiased estimator.

	\section{Asymptotic Analysis of the HCRB} \label{sec:HCRB_asymp}
	
	The expression of the exact HCRB computed in the previous section requires computing the inverse of matrix \eqref{eq:Fisher}, rendering its theoretical analysis and interpretation very challenging. In this section, an approximate expression of the HCRB for long range $\mathring{d}$ is presented. Then, the expression is further developed for the two cases of known and unknown shapes. Finally, some relationships are derived that link such results with the CRB of a point-like target.
	
	\subsection{ General result on long-range HCRB}
	
	We first give a general Theorem on the  approximation of the EFIM, useful for the subsequent derivation of the HCRB for both cases of known and unknown vehicle contour.
	
	\begin{theorem} \label{lemma:large_distance} For increasing ranges, the EFIM in \eqref{eq:Fisher} can be approximated as
		\begin{equation} \label{eq:large_distance}
		\mathbf{J}(\boldsymbol{\gamma}) = 
		2\frac{E}{N_0}
		\mathrlap{\underbrace{\phantom{\begin{bmatrix}\mathbf{T}_{11} & \mathbf{T}_{21}^\mathsf{T} \\ \mathbf{T}_{21} & \mathbf{T}_{22}	\end{bmatrix}}}_{\mathbf{T}}}
		\begin{bmatrix}
		\mathbf{T}_{11} & \mathbf{T}_{21}^\mathsf{T} \\
		\mathbf{T}_{21} & \mathbf{T}_{22}
		\end{bmatrix}
		+o\left(\mathring{d}^{-4}\right)
		\end{equation}
		where
		\begin{equation}
		\mathbf{T}_{11} = \begin{bmatrix}
		L & A & -A \\
		A & Z+B & -B \\
		-A & -B & B
		\end{bmatrix} \in \mathbb{R}^{3\times3} \label{eq:T11} \\[-0.2cm]
		\end{equation}
		\begin{align}
		A &=  L \langle\bar{w},\bar{w}\,\bar{\mathbf{p}}_\perp^\mathsf{T}\mathbf{R}\boldsymbol{\rho} \rangle_\star \\
	B &= L \left\|\bar{w}\,\bar{\mathbf{p}}_\perp^\mathsf{T}\mathbf{R}\boldsymbol{\rho}\right\|_\star^2
	+(\alpha+1)^2 \left\|\operatorname{P}_w^\perp(v)\right\|_\star^2 /\left\|w\right\|_\star^2  
		\end{align}
			where $\bar{w}=w/\|w\|_\star$ with $w$ given in \eqref{eq::defw},   $\bar{\mathbf{p}}_\perp = \mathbf{p}_\perp/\|\mathbf{p}_\perp\|_\star$ with $\mathbf{p}_\perp= \left(\begin{smallmatrix}0 & -1 \\ 1 & 0\end{smallmatrix}\right) \mathbf{p}$, $L$  given in \eqref{eq::defr},
		$Z$ given in \eqref{eq:a}, while
		$\mathbf{T}_{21} \in \mathbb{R}^{2Q\times3}$ and $\mathbf{T}_{22}\in\mathbb{R}^{2Q\times2Q}$ depend on the extended-target contour and are given in \eqref{eq::T21} and \eqref{eq::T22}, respectively.
	\end{theorem}
	
	\begin{proof}
		See Appendix~\ref{app:large_distance}.
	\end{proof}
	
	Theorem 2 provides a closed-form approximate version of the EFIM for long ranges, which admits a convenient block-structure representation in terms of three matrices $\mathbf{T}_{11}$, $\mathbf{T}_{12}$, and $\mathbf{T}_{22}$. The entries of such matrices still depend on the extended-target contour, being some of the involved terms functions of the variable $u$, as well as on the signal bandwidth (through $L$), array aperture (through $Z$), and target reflectivity (through $B$). Moreover, all the entries of the EFIM increase proportionally to the received energy $E$, which depends on $g^2\propto\mathring{d}^{-4}$. Thus, the entries of the EFIM decay as the fourth power of the range, and as we will demonstrate in the following propositions, the HCRB increases as the fourth power of the range as in the case of a point-like target in free space \cite[Eq. (1.6)]{skolnik1981introduction}.
	Next, explicit expressions for the HCRB of the range, direction and orientation are provided. Following the nomenclature of the previous section, we overload the notation and denote the HCRB related to the three parameter of interest as $\mathbf{C}(\mathring{d},\mathring{\phi},\varphi)$.  %We carry out the analysis for two different cases: (1) vehicle with known contour, and (2) vehicle with unknown contour, because their HCRB vary substantially analytically and numerically.

	\subsection{HCRB for Known Shape} \label{sub:known_contour}
	 The result of Theorem~\ref{lemma:large_distance} is now further developed for the case of known vehicle contour.
	
	\begin{proposition} \label{lemma:far_known}
		For an extended target with \emph{known contour}, i.e., vectors $\mathbf{m}$ and $\mathbf{n}$ known,
		\begin{equation} \label{eq:far_known}
		\mathbf{C}(\mathring{d},\mathring{\phi},\varphi) = \left(2\frac{E}{N_0}\right)^{-1} \mathbf{T}_{11}^{-1}
		+o\left(\mathring{d}^{4}\right)
		\end{equation}
		where
		\begin{equation} \label{eq:T11:inv}
		\mathbf{T}_{11}^{-1} \!=\!
		\begin{bmatrix}
		\left(L -\frac{A^2}{B}\right)^{-1} & 0 & -\left(A-\frac{LB}{A}\right)^{-1} \\
		0 & Z^{-1} & Z^{-1} \\
		-\left(A-\frac{LB}{A}\right)^{-1} & Z^{-1} & Z^{-1}+\left(B-\frac{A^2}{L}\right)^{-1}
		\end{bmatrix}.\!
		\end{equation}
	\end{proposition}
	
	\begin{proof}
		See Appendix~\ref{app:known_contour}.
	\end{proof}

	Proposition 1 provides a closed-form expression for the HCRB in case of known target contour, which mainly depends on the inverse of the matrix $\mathbf{T}_{11}$. More specifically, the elements in the main diagonal of $\mathbf{T}^{-1}_{11}$ unveil the main connections between the accuracy in the estimation of $\mathring{d}$, $\mathring{\phi}$, and $\varphi$ and the fundamental radar parameters. By comparing the CRB of the point-like target \eqref{eq:Fisher:point_target} and the HCRB of the extended target with known contour in Proposition~\ref{lemma:far_known} for the same received energy, we discover that  $C(\mathring{\phi}) \approx C_\text{p.t.}(\mathring{\phi}) = (2E/N_0)^{-1}Z^{-1}$ for a sufficiently large range. On the other hand, the HCRB on the range for an extended target is larger than the CRB of the point-like target because $C(\mathring{d}) \approx (2E/N_0)^{-1}(L -A^2/B)^{-1} \geq C_\text{p.t.}(\mathring{d})$.
	By using the explicit expressions of $A$ and $B$, we find that $C(\mathring{d}) = (2E/N_0)^{-1}(L -A^2/B)^{-1}$ is bounded as
	\begin{equation} \label{eq:star_product_approx}
	\frac{N_0}{2E}L^{-1} \!\!\leq
	C(\mathring{d}) \! < \!
	\frac{N_0}{2E}L^{-1}\! \left(\!1 -\left[\frac{\langle w,w\,\mathbf{p}_\perp^\mathsf{T}\mathbf{R}\boldsymbol{\rho} \rangle_\star}{\|w\|_\star \|w\,\mathbf{p}_\perp^\mathsf{T}\mathbf{R}\boldsymbol{\rho}\|_\star}\right]^2 \right)^{-1}
	\!\!\!\!\!\!\approx  L^{-1 }
	\end{equation}
	because $\langle w,w\,\mathbf{p}_\perp^\mathsf{T}\mathbf{R}\boldsymbol{\rho} \rangle_\star \|w\|_\star^{-1} \|w\,\mathbf{p}_\perp^\mathsf{T}\mathbf{R}\boldsymbol{\rho}\|_\star^{-1} \approx 0$. The reason is that $w$ is positive by definition whereas $w\,\mathbf{p}_\perp^\mathsf{T}\mathbf{R}\boldsymbol{\rho}$ is an odd function around some point $u$, and so for most contours %we have verified numerically that
	the star-product is close to zero. Thus, $C(\mathring{d}) \approx C_\text{p.t.}(\mathring{d})$.
	The orientation is not defined for the case of a point-like target. For the extended target, we can rewrite $C(\varphi) \approx (2E/N_0)^{-1}[Z^{-1} +(B-A^2/L)^{-1}] \approx (2E/N_0)^{-1} [Z^{-1} +(L/B)(L-A^2/B)^{-1}]$, and by the same approximation, $C(\varphi) \approx (2E/N_0)^{-1} [Z^{-1} +B^{-1}]$. Hence, $C(\varphi) \approx C(\mathring{\phi}) +(2E/N_0)^{-1} B^{-1}$, whereas if we were to compute the HCRB of the relative orientation\footnote{We define the relative orientation as the heading angle of the target with respect to the line connecting the radar and the target; the absolute orientation (or simply orientation) is instead the heading angle with respect to the $x$-axis.}  (related to $\varphi$ through $\varphi = \varphi_\text{rel.} +\mathring{\phi}$), we would obtain $C(\varphi_\text{rel.}) \approx (2E/N_0)^{-1} B^{-1}$. This indicates that the variance of the absolute orientation is the sum of two variances: the direction variance and the relative orientation variance. In summary, we have proven the following result.
	
	\begin{proposition}
	For an extended target with \emph{known contour} at a sufficiently long range, the following approximate relationships hold true\vspace{-0.1cm}
	\begin{align}
	C(\mathring{d}) &\approx C_\text{p.t.}(\mathring{d}) \label{eq:HCRB:known:range} \\
	C(\mathring{\phi}) &\approx C_\text{p.t.}(\mathring{\phi}) \\
	C(\varphi) &\approx C_\text{p.t.}(\mathring{\phi}) +(2E/N_0)^{-1} B^{-1}. \label{eq:HCRB:known:orientation}
	\end{align}
	\end{proposition}
	
	\subsection{HCRB for Unknown Shape}\label{sec::HCRB_unknown}
	
		\begin{table*}
\caption{Summary of Cram\'er-Rao bound expressions for point-like and extended targets}\label{table}
\begin{center}
\begin{tabular}{c|c}
\hline\hline\\[-0.1cm]
Exact HCRB  & $ \frac{\|w\|_\star^2}{2} \left(\frac{2 E}{N_0}\right)^{-1}
		\Big[
		 L
		\langle w\boldsymbol{\mu},w\boldsymbol{\mu} \rangle_\star
		+ M
		\langle w\cos\phi\boldsymbol{\eta},w\cos\phi\boldsymbol{\eta} \rangle_\star
		+ (\alpha+1)^2
		\langle \operatorname{P}_w^\perp\left(v\,\boldsymbol{\xi}\right), \operatorname{P}_w^\perp\left(v\,\boldsymbol{\xi}\right) \rangle_\star  \Big]^{-1}$\\[0.2cm]
\hline\\[-0.1cm]
CRB for point-like target & $\mathbf{C}_\text{p.t.}(\mathring{d},\mathring{\phi}) = \begin{bmatrix}
C_\text{p.t.}(\mathring{d}) & 0 \\
0 & C_\text{p.t.}(\mathring{\phi})
\end{bmatrix}$ \\[0.2cm]
& with $C_\text{p.t.}(\mathring{d}) = \left(\frac{2 E}{N_0}\right)^{-1} L^{-1}$, \quad  $C_\text{p.t.}(\mathring{\phi}) = \left(\frac{2 E}{N_0}\right)^{-1} Z^{-1}$
\\[0.2cm]
\hline\\[-0.1cm]
Approximate HCRB, known shape\\[-0.4cm] (for long range) & $\mathbf{C}(\mathring{d},\mathring{\phi},\varphi) \approx \left(2\frac{E}{N_0}\right)^{-1} \begin{bmatrix}
		\left(L -\frac{A^2}{B}\right)^{-1} & 0 & -\left(A-\frac{LB}{A}\right)^{-1} \\
		0 & Z^{-1} & Z^{-1} \\
		-\left(A-\frac{LB}{A}\right)^{-1} & Z^{-1} & Z^{-1}+\left(B-\frac{A^2}{L}\right)^{-1}
		\end{bmatrix}$  \\[0.8cm]
& $C(\mathring{d}) \approx C_\text{p.t.}(\mathring{d})$, \quad
 $C(\mathring{\phi}) \approx C_\text{p.t.}(\mathring{\phi})$\\[0.2cm]
& $C(\varphi) \approx C_\text{p.t.}(\mathring{\phi}) +\left(2\frac{E}{N_0}\right)^{-1} B^{-1}$
\\[0.2cm]
\hline\\[-0.1cm]
Approximate HCRB, unknown shape\\[-0.2cm] (for long range) & $\mathbf{C}(\mathring{d},\mathring{\phi},\varphi) \approx \left(2\frac{E}{N_0}\right)^{-1} \begin{bmatrix}
		C^{-1} & 0 & -D^{-1} \\
		0 & Z^{-1} & Z^{-1} \\
		-D^{-1} & Z^{-1} & Z^{-1}+F^{-1}
		\end{bmatrix}$  \\[0.6cm]
& $C(\mathring{d}) \approx  C_\text{p.t.}(\mathring{d}) \,
	 \|w\|_\star^{2} / \left\| \operatorname{P}_\mathcal{B}^\perp(w,0) \right\|_\star^2$, \quad  $C(\mathring{\phi}) \approx C_\text{p.t.}(\mathring{\phi})$\\[0.2cm]
& $C(\varphi) \approx C_\text{p.t.}(\mathring{\phi}) +\left(2\frac{E}{N_0}\right)^{-1}	 
	 \|w\|_\star^{2} / \left\|\operatorname{P}_\mathcal{D}^\perp\left(\sqrt{L}\,w\,\bar{\mathbf{p}}_\perp^\mathsf{T}\mathbf{R}\boldsymbol{\rho},(\alpha+1)\operatorname{P}_w^\perp (v)\right) \right\|_\star^2
	$ \\[0.3cm]
\hline\hline
\end{tabular}
\end{center}
\label{tabella}
\end{table*}%

		The result of Theorem~\ref{lemma:large_distance} is now further developed for the case of unknown vehicle contour.
		
	\begin{proposition} \label{lemma:far_unknown}
		For an extended target with \emph{unknown contour}, i.e., vectors $\mathbf{m}$ and $\mathbf{n}$ unknown, 
		\begin{equation} \label{eq:far_unknown}
		\mathbf{C}(\mathring{d},\mathring{\phi},\varphi) = \left(2\frac{E}{N_0}\right)^{-1} \mathbf{U}
		+o\left(\mathring{d}^{4}\right)
		\end{equation}
		where\vspace{-0.2cm}
		\begin{align}
		\mathbf{U} &=
		\begin{bmatrix}
		C^{-1} & 0 & -D^{-1} \\
		0 & Z^{-1} & Z^{-1} \\
		-D^{-1} & Z^{-1} & Z^{-1}+F^{-1}
		\end{bmatrix}\\
		C &= L-H-\frac{(A-J)^2}{B-I}\\
		D &= -A-J-\frac{(L-H)(B-I)}{A-J} \\
		F &= B-I-\frac{(A-J)^2}{L-H},
		\end{align}
		$H = \|\mathbf{T}_{22}^{-\nicefrac{1}{2}}\mathbf{c}\|^2$, $I = \|\mathbf{T}_{22}^{-\nicefrac{1}{2}}\mathbf{d}\|^2$, $J = \mathbf{c}^\mathsf{T} \mathbf{T}_{22}^{-1} \mathbf{d}$, and $\mathbf{T}_{22}$ is provided in \eqref{eq::T22}, while $\mathbf{c}$ and $\mathbf{d}$ are given in \eqref{eq::cbold} and \eqref{eq::dbold}.
	\end{proposition}
	
	\begin{proof}
		See Appendix~\ref{app:unknown_contour}.
	\end{proof}

	As for the case of known contour, the HCRB on the direction converges to the CRB of the point-like target as the range tends to infinity, namely $C(\mathring{\phi})\approx C_\text{p.t.}(\mathring{\phi})= (2E/N_0)^{-1}Z^{-1}$.
	The expressions of the HCRB on the range and orientation are considerably more complicated. To this end, we first define the new inner product: $\langle(f_1,f_2),(g_1,g_2) \rangle_\star \triangleq \langle f_1,g_1 \rangle_\star +\langle f_2,g_2 \rangle_\star$ over the space $\mathcal{F}\times\mathcal{F}$ where $\mathcal{F}=L^2([0, 2\pi])$. For convenience, the new inner product is also overloaded such that $\langle(\mathbf{f}_1,\mathbf{f}_2),(\mathbf{g}_1,\mathbf{g}_2) \rangle_\star \triangleq \langle\mathbf{f}_1,\mathbf{g}_1 \rangle_\star +\langle\mathbf{f}_2,\mathbf{g}_2 \rangle_\star$. The associated projection operator is $\operatorname{P}_\mathcal{A}(f_1,f_2) \triangleq {\langle(f_1,f_2),(\mathbf{g}_1,\mathbf{g}_2) \rangle_\star} {\langle(\mathbf{g}_1,\mathbf{g}_2),(\mathbf{g}_1,\mathbf{g}_2) \rangle_\star^{-1}} (\mathbf{g}_1,\mathbf{g}_2)$ where the rows of $(\mathbf{g}_{1},\mathbf{g}_{2})$ form a basis spanning $\mathcal{A}\subset\mathcal{F}\times\mathcal{F}$, and $\operatorname{P}_\mathcal{A}^\perp(f_1,f_2) \triangleq (f_1,f_2)-\operatorname{P}_\mathcal{A}(f_1,f_2)$. Then, we obtain
	\begin{align}
	L - H & =  \|w\|_\star^{-2} \left\|  \operatorname{P}_\mathcal{A}^\perp(\sqrt{L}w,0) \right\|_\star^2 \\
	B  - I & =  \|w\|_\star^{-2} \left\|\operatorname{P}_\mathcal{A}^\perp\left(\sqrt{L}\,w\,\bar{\mathbf{p}}_\perp^\mathsf{T}\mathbf{R}\boldsymbol{\rho},(\alpha+1)\operatorname{P}_w^\perp (v)\right) \right\|_\star^2 \\
	A - J & =  \|w\|_\star^{-2} \nonumber \\
	&  \times \left\langle \operatorname{P}_\mathcal{A}^\perp(\sqrt{L}w,0), \operatorname{P}_\mathcal{A}^\perp\left( \sqrt{L}\,w\,\bar{\mathbf{p}}_\perp^\mathsf{T}\mathbf{R}\boldsymbol{\rho},(\alpha + 1)\operatorname{P}_w^\perp (v)\right) \right \rangle_{\! \star}
	\end{align}
	where $\mathcal{A}=\operatorname{span}\Big\{\{(\sqrt{L}ws_q, (1+\alpha)t_q)\}_{q=1}^{2Q}\Big\}$, and $s_q$ and $t_q$ are the components of vectors \eqref{eq:g} and \eqref{eq:h}, respectively. We have therefore proven the following result.
	\begin{proposition}
	    For an extended target with \emph{unknown contour} at a sufficiently long range, the following approximate relationships hold true
	\begin{align}
	C(\mathring{d}) &\approx C_\text{p.t.}(\mathring{d})
	\left(\|w\|_\star^{-2} \left\| \operatorname{P}_\mathcal{B}^\perp(w,0) \right\|_\star^2\right)^{-1} \label{eq:HCRB:unknown:range} \\
	C(\mathring{\phi}) &\approx C_\text{p.t.}(\mathring{\phi}) \\
	\begin{split}
	C(\varphi) &\approx C_\text{p.t.}(\mathring{\phi}) +\left(2\frac{E}{N_0}\right)^{-1}	 \\
	&  \times \left(\|w\|_\star^{-2} \left\|\operatorname{P}_\mathcal{D}^\perp\left(\sqrt{L}\,w\,\bar{\mathbf{p}}_\perp^\mathsf{T}\mathbf{R}\boldsymbol{\rho},(\alpha+1)\operatorname{P}_w^\perp (v)\right) \right\|_\star^2
	\right)^{-1}
	\end{split} \label{eq:HCRB:unknown:orientation}
	\end{align}
	where $\mathcal{B}=\operatorname{span}\{\mathcal{A},(\sqrt{L}w\bar{\mathbf{p}}_\perp^\mathsf{T}\mathbf{R}\boldsymbol{\rho},(1+\alpha)\operatorname{P}_w^\perp (v))\}$ and $\mathcal{D}=\operatorname{span}\{\mathcal{A},(\sqrt{L}w,0)\}$. 
		\end{proposition}

	The results of the whole section are summarized in Table~\ref{table}, from which a few remarks follow. First, we notice that in the approximate HCRB for known and unknown shapes, the bound on the direction $\mathring{\phi}$ is always equal to that of the point-like target case since for long ranges this parameter becomes independent of the shape. %Conversely, for both known and unknown shape the bounds on the orientation $\varphi$ include additive positive terms, that produce a performance degradation compared to the case of point-like target.
	Finally, as to the distance $\mathring{d}$, there is a difference between known and unknown shape, with the former exhibiting the same bound of the point-like target case while the latter showing an amplification factor that increases the value of the bound. These observations reveal how the lack of knowledge on the vehicle contour impacts onto the achievable estimation performance.

	\section{Numerical Results} \label{sec:simulations}
	
	\subsection{Accuracy Vs.\ Range}
	\label{sub:asymptotic_analysis:results}
	
	This section analyzes the HCRB for the studied cases of known and unknown contour, as a function of the range. The exact HCRB for the three parameters of interest (range, direction and orientation)  is numerically computed by inverting the EFIM \eqref{eq:Fisher}. For comparison purposes, we also plot the CRB of a point-like target for the range and direction assuming equal received energy, together with the long-range approximations of the HCRB for known contour \eqref{eq:HCRB:known:range}--\eqref{eq:HCRB:known:orientation} and unknown contour \eqref{eq:HCRB:unknown:range}--\eqref{eq:HCRB:unknown:orientation}.
	% and the lower-bound of \eqref{eq:HCRB_range:lb}.
	\begin{figure}
		\centering
		\subfloat[]{
			\includegraphics[width=0.85\columnwidth]{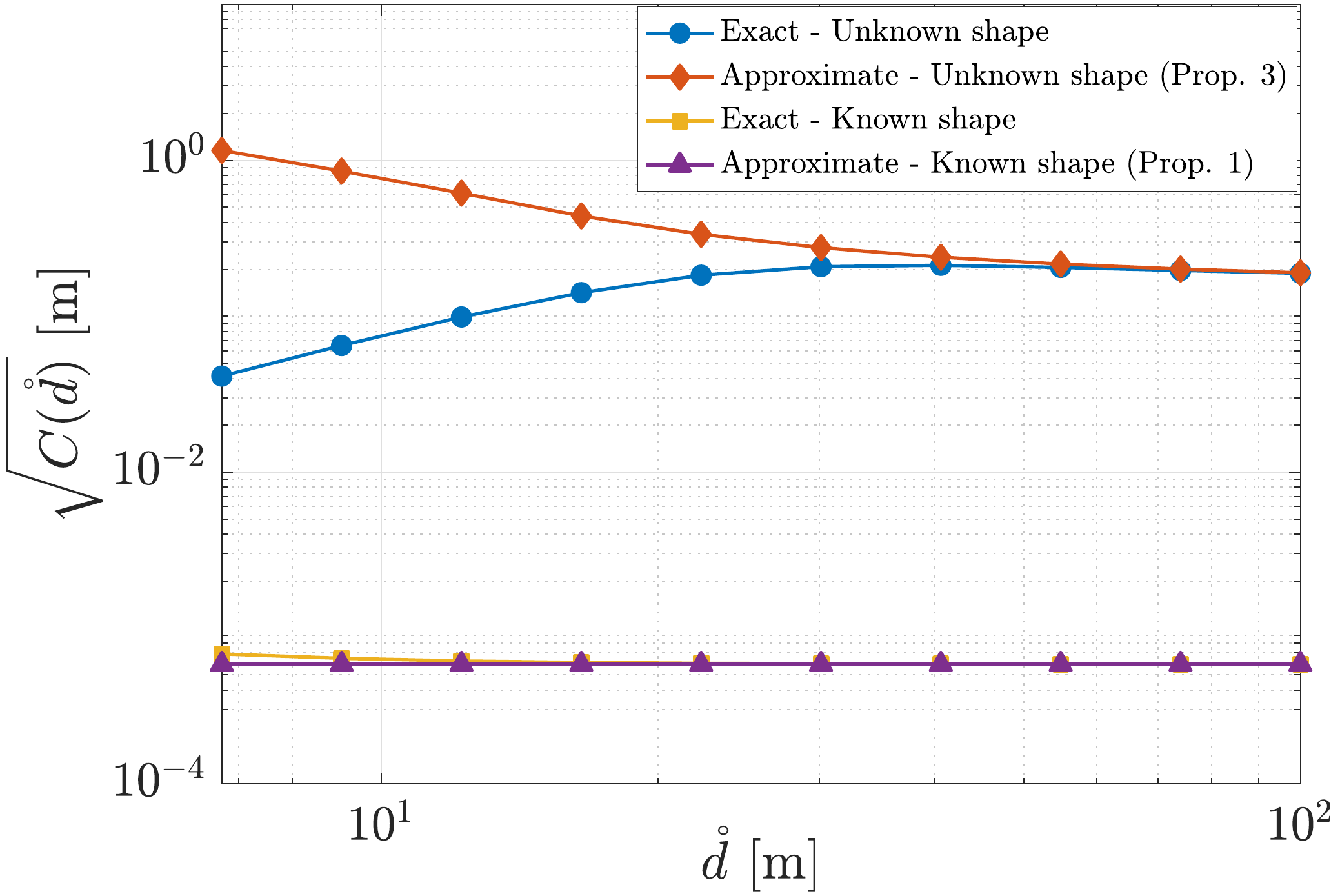}
			\label{fig:HCRBrange} }
			\vspace{-0.2cm}
		\\
		\subfloat[]{
			\includegraphics[width=0.85\columnwidth]{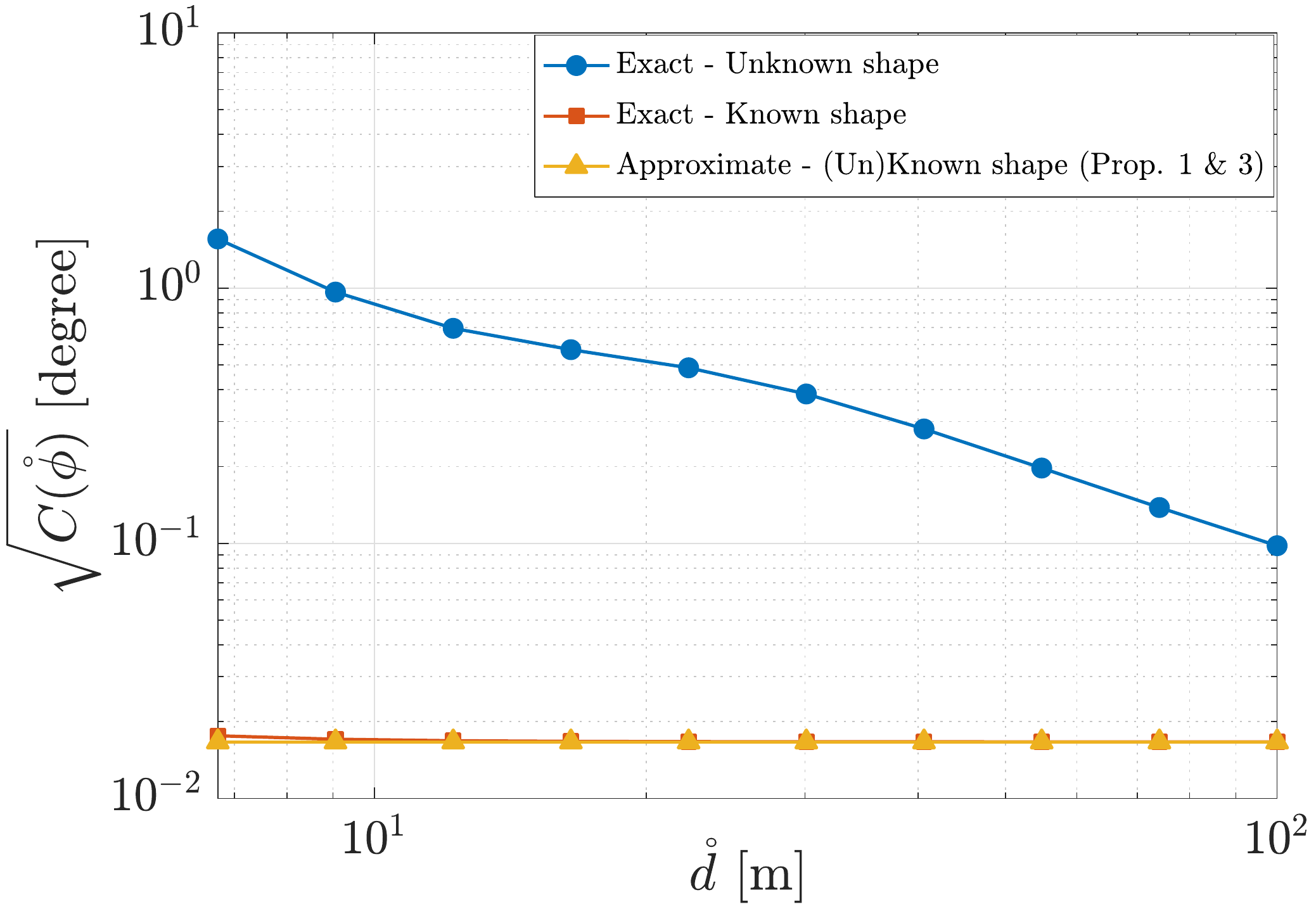}
			\label{fig:HCRBdirection} }
			\vspace{-0.2cm}
		\\
		\subfloat[]{
			\includegraphics[width=0.85\columnwidth]{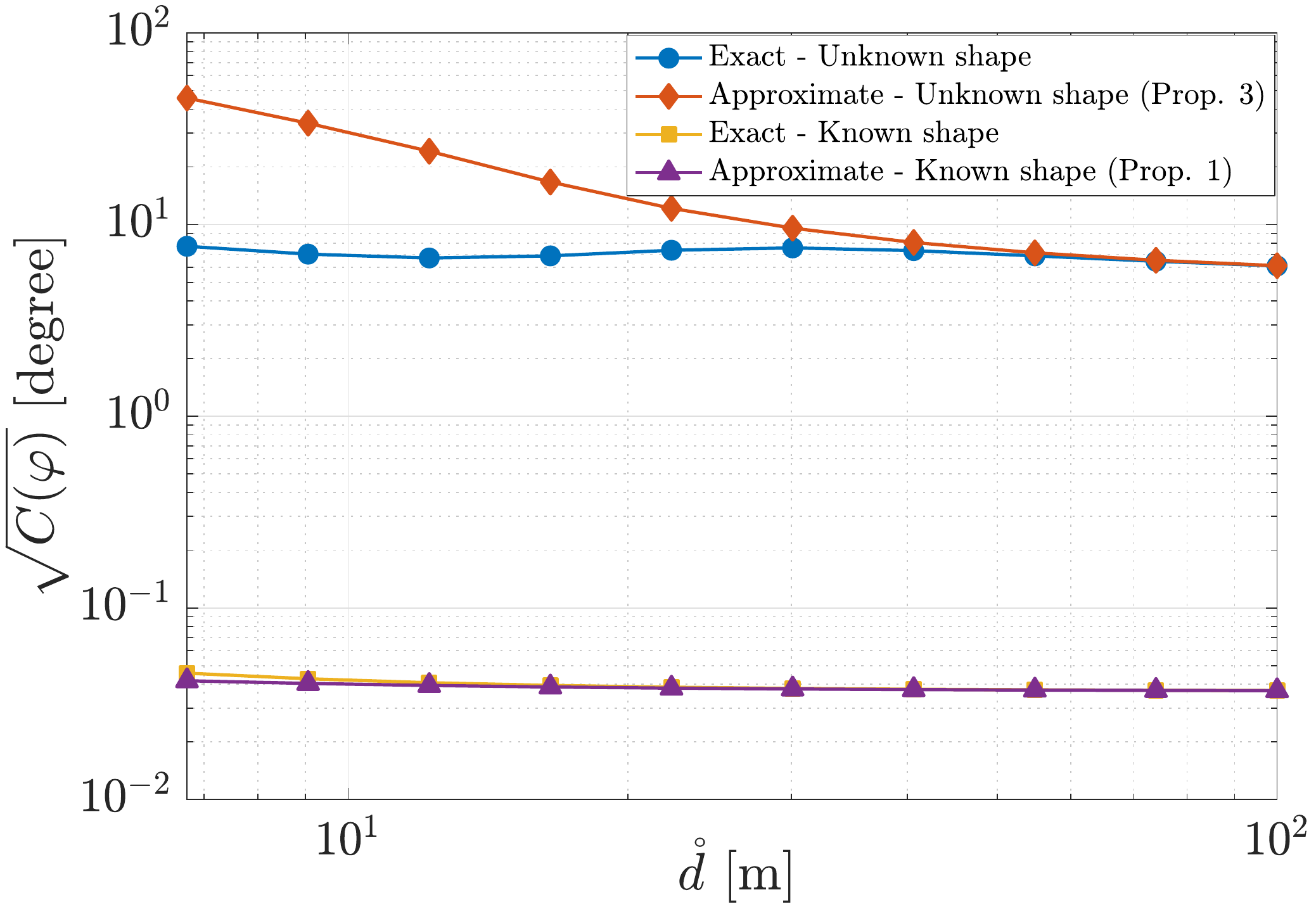}
			\label{fig:HCRBorientation} }
		\caption{HCRBs of the parameters of interest (range, direction and orientation) versus distance. 
		%{\color{red}These plots will include an extra lines of labels below the $x$-axis with the value of $E/N_0$ once converted to PGFplots. Not easy to do in Matlab.}
		}
		\label{fig:HCRB}
	\end{figure}
	For all plots, the radar is kept at the origin of the global reference system with its broadside pointing the positive $x$-axis and the vehicle's heading is pointing towards the positive half-plane of the $y$-axis, so that $\varphi = \pi/2$. The contour of the vehicle is the same as in Fig.~\ref{fig:geometric_model} and has a  length of approximately \SI{11.2}{\meter}, with a reflectivity coefficient set to $\alpha=5$, and it is parameterized by $Q = 10$ harmonics. The radar transmits a standard chirp signal $s(t) = \frac{1}{\sqrt{T}}\mathrm{rect}(\frac{t}{T})e^{j\pi\frac{B}{T}t^2}$, with $\mathrm{rect}(\cdot)$ the rectangular pulse in the interval $[-1/2,1/2]$, having bandwidth $B=\SI{1}{\giga\hertz}$ and duration $T=\SI{10}{\micro\second}$, and is equipped with a ULA composed of $N = 30$ antennas with half-wavelength spacing. The carrier frequency has thus no impact on the baseband signal, hence is left unspecified.

		\begin{figure}
		\vspace{-0.3cm}
		\centering
		\hspace{-0.6cm}\subfloat[]{
			\includegraphics[width=.85\columnwidth]{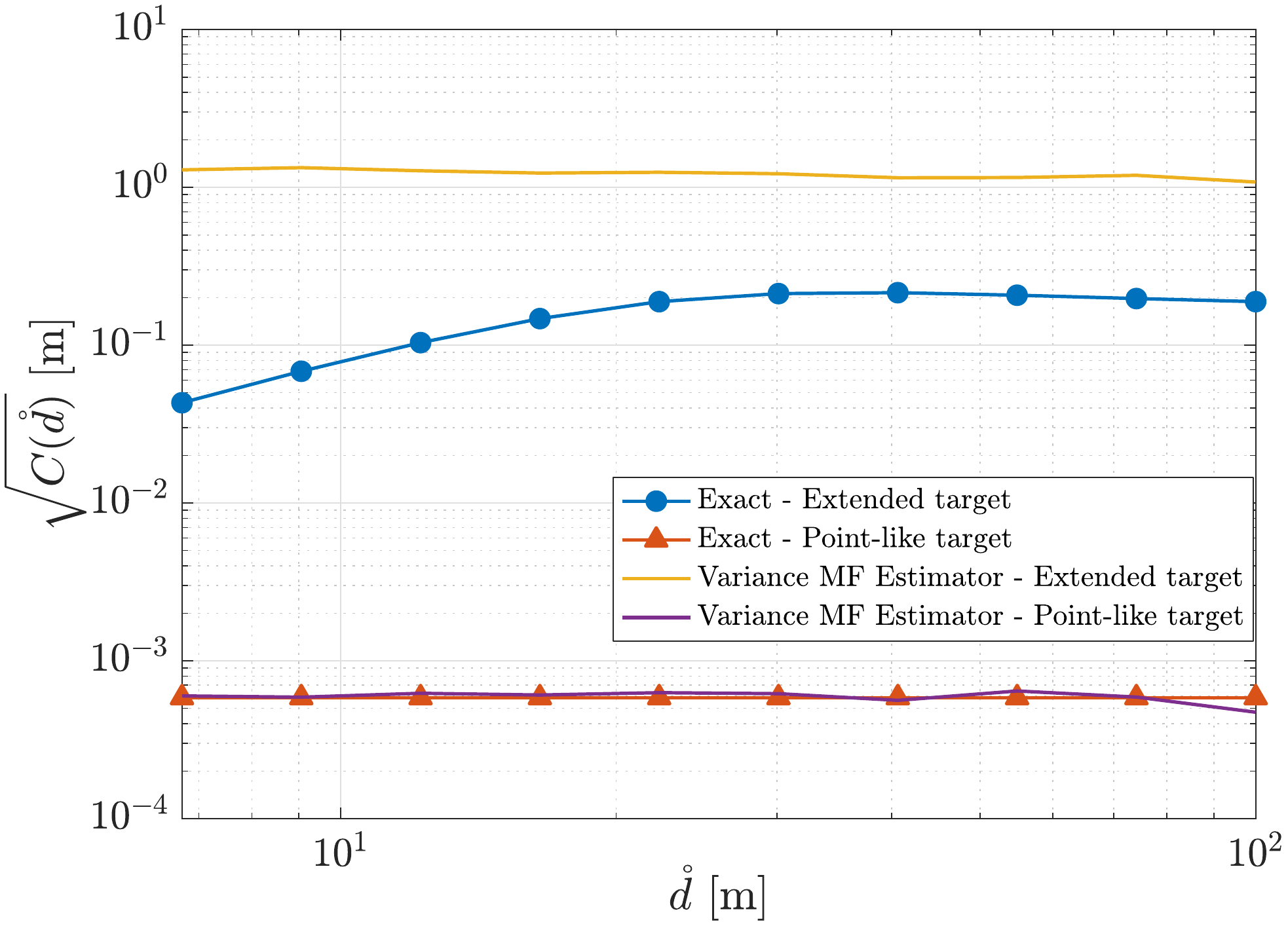}
			\label{fig:tightness:range} }
			\vspace{-0.2cm}
		\hfil
			\subfloat[]{\hspace{-0.6cm}
			\includegraphics[width=.85\columnwidth]{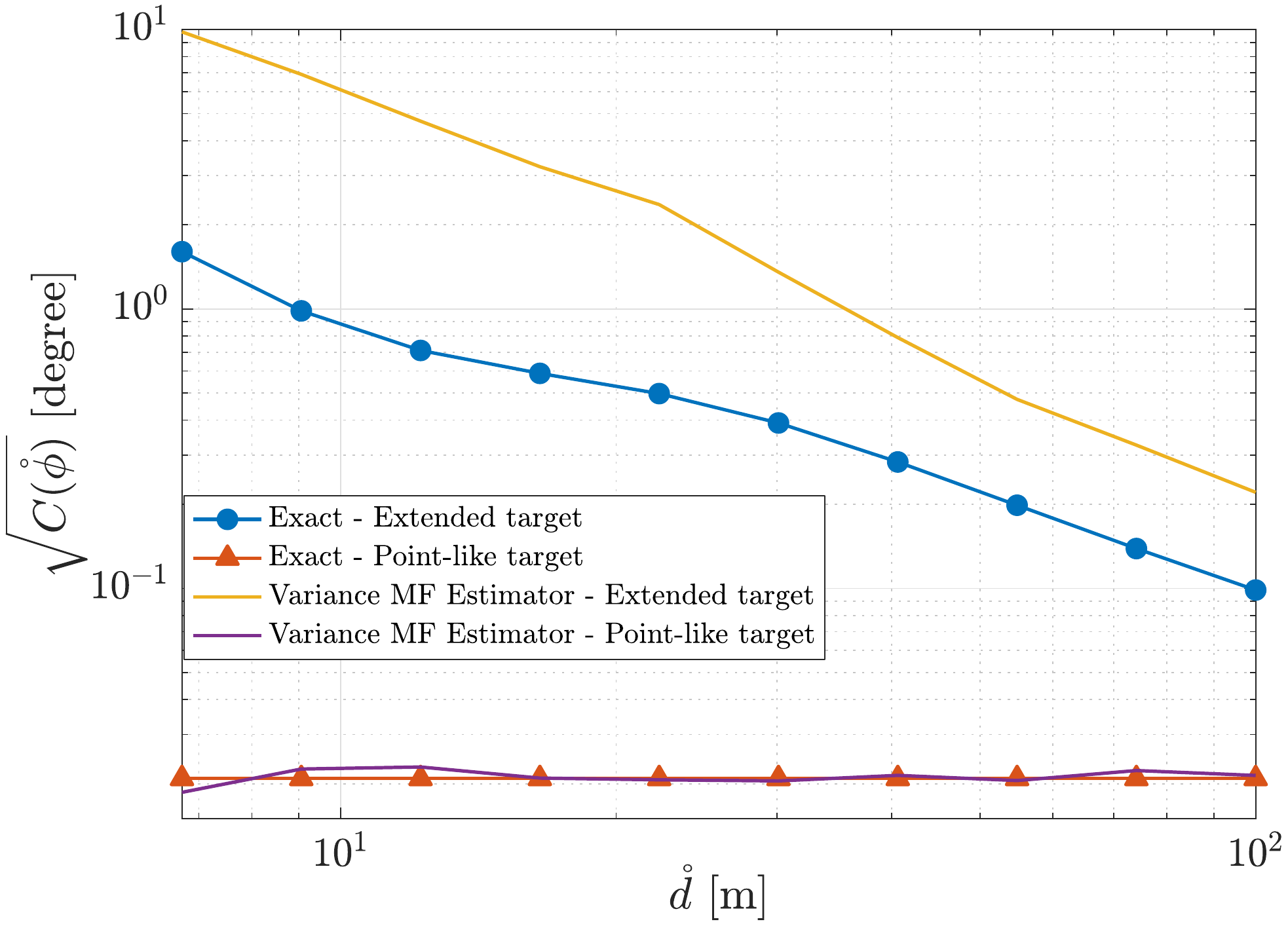}
			\label{fig:tightness:direction} }
		\caption{Variance of the MF-based estimator compared to the HCRBs of range and direction parameters as a function of the distance between target and radar. }
		\label{fig:tightness}
	\end{figure}
	
	In Fig.~\ref{fig:HCRB}, we report the HCRB on the parameters of interest as a function of the target vehicle's position $\mathbf{p}$ ranging from $[6 \ 3]^\mathsf{T}$ to $[89\ 45]^\mathsf{T}$. For the sake of the analysis, the energy is kept fixed\footnote{The dependence of the HCRBs on the energy $E$ is quite intuitive since all the bounds will scale accordingly.} (and set set such that $E/N_0=\SI{40}{\decibel}$ at $\mathbf{p}=[6\ 3]^\mathsf{T}$) irrespective of the distance between the radar and the target, so that any change of the HCRBs can be attributed to the sole range.
	%All the HCRBs have a positive slope because the received energy, which has a direct effect on the HCRB, decreases as the fourth power of the range.
	In case of known contour, we find that for the range and direction parameters, the approximate expressions of the HCRBs practically match the exact HCRBs over all the considered ranges, confirming the validity of the results obtained in Section~\ref{sub:known_contour}. Moreover, since the approximate HCRBs equal the CRBs of a point-like target, this analysis reveals that the latter can be considered a good lower bound for the variance of any unbiased estimator of range and direction when the contour of the extended target is known.
	
	%the HCRB and its approximation given in Proposition \ref{lemma:far_known} practically match the CRB of the point-like target over all the considered ranges, confirming the analysis of Section~\ref{sub:known_contour}. Thus, for practical purposes the CRB of a point-like target can be considered  a good lower-bound for the variance of any unbiased estimator of the range and direction when the contour of the extended-target is known.
	When the contour of the target is unknown, the HCRB on all parameters of interest is around three orders of magnitude worse, with a subtle difference for the direction parameter $\mathring{\phi}$, whose HCRBs for known and unknown contour become closer as the range increases (not observable in Fig.~\ref{fig:HCRBdirection}). This behavior is linked to the fact that, for large distances,  the extended target practically degenerates into a point-like target regardless of its contour. The HCRBs on the range and orientation for unknown contour converge to their asymptotic approximations at about \SI{40}{\meter} of distance. Nonetheless, it is clear from the figure that knowing the contour of the extended target is critically important for its accurate localization. The HCRB on the orientation is always larger than the HCRB on the direction, confirming the validity of \eqref{eq:HCRB:known:orientation} and \eqref{eq:HCRB:unknown:orientation}. The range at which the asymptotic HCRB of the range and orientation converge to the true HCRB is smaller for a reduced number of coefficients $Q$ in the vehicle contour: in the extreme case of $Q=1$ (not reported here) such a range is as small as \SI{5}{\meter} for all the three parameters. Notice that $Q = 1$ is tantamount to considering a target contour approximated by an ellipse with unknown semi-axes; intuitively, it can be interpreted as a ``high-level" model to be adopted when the fine-grained details of the target are not of strict interest, or if the available bandwidth is not large enough to observe them.
	
	The next two sections study (i) the validity of the HCRB as a tight lower bound, and (ii) the benefits of multiple radars sensing the same target. The simulation parameters are the same  described in the first paragraph of this section.

	\subsection{Validation of the HCRB}

	Although the HCRB is a lower bound on the variance of any unbiased estimator, it is not necessarily tight asymptotically at high $E/N_0$ as is the marginal CRB \cite{noam2009notes}. 
	Developing an optimal unbiased estimator of the range, direction and orientation for extended targets  is an open-ended problem to the best of our knowledge. For comparison purposes, a simple estimator of the range and direction is presented next and applied to the extended and point targets.
	%We simulate the received signal assuming the radar waveform is a chirp sequence of duration $T=\SI{10}{\micro\second}$ and bandwidth $B=\SI{1}{\giga\hertz}$.
	Let $\mathbf{Y} = [\mathbf{y}(0) \ \cdots \ \mathbf{y}((W-1)/f_\text{s})]$ be the baseband signal at the receiver \eqref{eq:received_signal} plus noise after low-pass filtering and sampling, with $W$ the number of samples and $f_\text{s}$ the sampling frequency.  Then, the proposed estimator of the direction is
	\begin{equation} \label{eq:estimator:direction}
	\hat{\phi} = \argmax{\phi} \left\|\mathbf{a}^\mathsf{H}(\phi) \mathbf{Y}\right\|.
	\end{equation}
	Once the direction is estimated, we perform coherent integration of the  received signals  $\mathbf{z} = \mathbf{a}^\mathsf{H}(\hat{\phi}) \mathbf{Y}$ and apply the standard de-chirping plus FFT estimator in FMCW radars \cite{lees1989digital}. More specifically, we first multiply the known transmitted waveform $s(t)$ with the  received signals (matched filtering). Denoting with $\bm{s} = [s(0) \ \cdots \ s((W-1)/f_s)]^\mathsf{T}$ the vector containing $W$ samples of the transmitted pulse, the mixed received signals can be thus obtained as $\bm{z}_\text{mix}(\hat{\phi}) = \bm{s}^* \odot \bm{z}(\hat{\phi)}$, with $\odot$ denoting the Hadamard element-wise product operator. To obtain an estimate of the range $d$, we exploit the fact that the elements of $\bm{z}_\text{mix}(\hat{\phi})$ can be interpreted as discrete samples of a complex exponential with  frequency $\nu = \frac{2Bd}{T f_s c}$. Accordingly, an estimate of $d$ can be obtained by searching for the frequency $\hat{\nu}$ corresponding to the dominant peak in the FFT of the vector $\bm{z}_\text{mix}(\hat{\phi})$ and by  reversing the relationship as
\begin{equation}
\hat{d} = \frac{\hat{\nu} f_s T c}{2B}. \label{eq:estimator:range}
\end{equation}
Given its structure, in the following we denote \eqref{eq:estimator:direction} and \eqref{eq:estimator:range} as a matched-filter (MF) based estimator.
	
	The variance of the estimators is estimated via a Monte Carlo simulation based on 100 independent trials and it is plotted in Fig.~\ref{fig:tightness} together with the HCRB of the extended target and the CRB of the point-like target.
	It is worth noting that the estimator of the direction \eqref{eq:estimator:direction} does not exploit a priori knowledge of neither the waveform nor the target contour, while the estimator of the range assumes a single return with the same waveform of the transmitted signal.
	Despite being suboptimal approaches, the proposed estimators work remarkably well for the case of a  point-like target. When applied to the extended target, a gap is observed with the theoretical lower bounds, which however reduces as  the range between the radar and the target increases, confirming the validity of the derived HCRBs. More sophisticated estimators (which are beyond the scope of the present contribution) would probably help bridging the gap towards the HCRB.

	\subsection{Radar Diversity}
	\begin{figure}
		\centering
		\includegraphics[width=0.8\columnwidth]{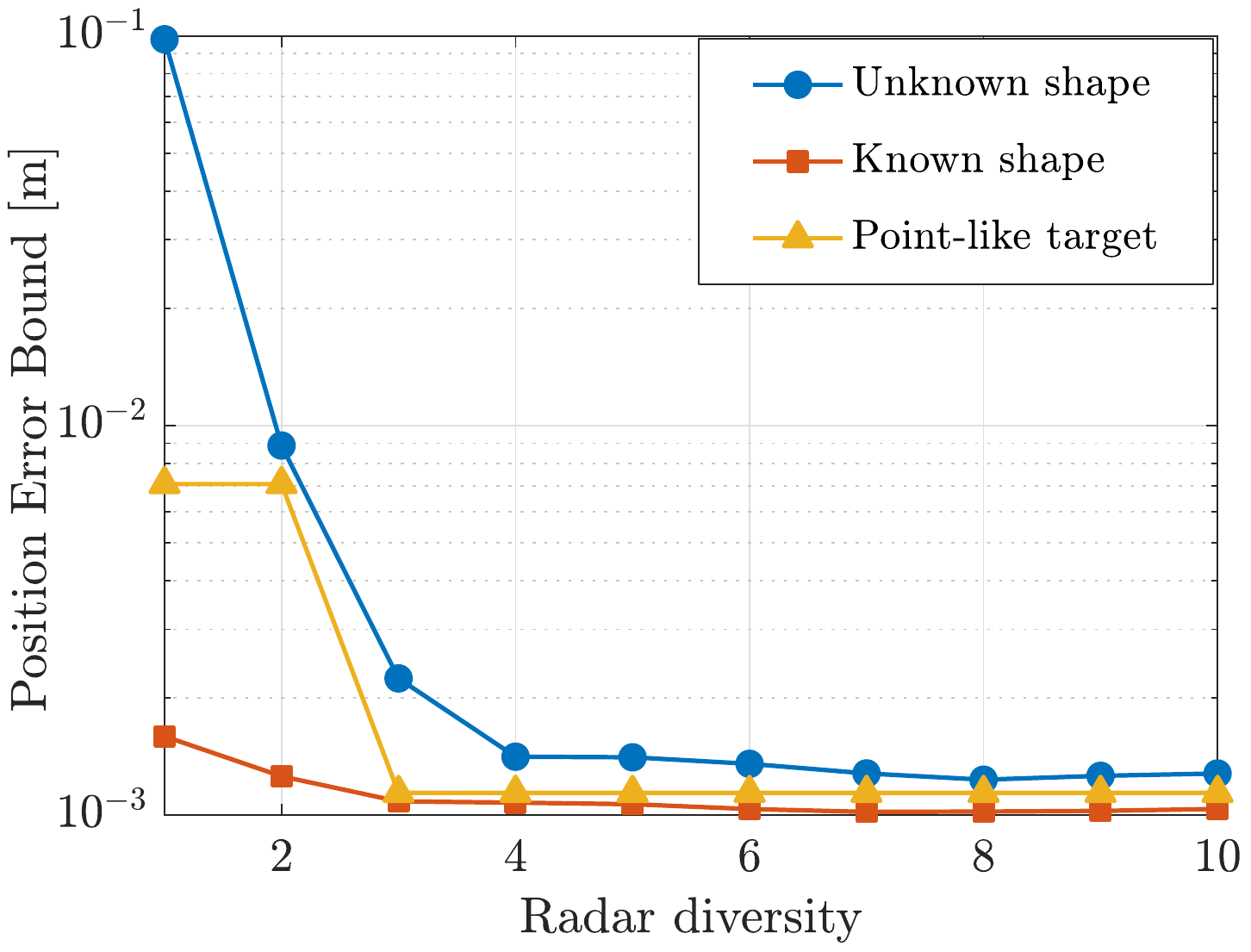}
		\caption{PEB as a function of the number of radars sensing the target.}
		\label{fig:diversity}
	\end{figure}	
	
	 In this section, we evaluate the HCRB when multiple radars are uniformly located around the target vehicle at a fixed range of \SI{7}{\meter}, all of them with the broadside of the array pointing towards the vehicle.
	To keep the analysis fair, the aggregated received energy by all radars remains constant, i.e., $K_rE/N_0=\SI{40}{\decibel}$ where $K_r$ is the number of radars, so that any change in terms of HCRB can be solely attributed to the radar diversity.
	Fig.~\ref{fig:diversity} plots the PEB for an increasing number of radars. For the case of known contour, the positioning accuracy (PEB) improves slightly and saturates already for 3 radars. On the other hand, when the contour is unknown, the positioning accuracy significantly improves as the number of radars increases. This behavior is linked to the fact that the vehicle is sensed from different angles, and consequently the contour can be more accurately estimated. Indeed, by simply passing from 1 to 2 radars, the PEB decreases by about one order of magnitude.
	When 4 or more radars sense the vehicle,  the PEB for unknown contour reduces to about  twice the value of the PEB in case of known contour.

	\section{Conclusions}\label{sec::conclusions}
	This paper investigated the theoretical accuracy achievable in the estimation of range, direction and orientation information of a radar sensing a vehicle modeled as an extended target. We have developed  analytical models that correctly capture the behavior of an extended target and lend themselves to be used for analyzing the impact of the different signal and scenario parameters. The exact HCRB provided in Theorem \ref{prop1} revealed that the ultimate accuracy depends on three different terms related to the square of the signal effective bandwidth, the square of the array aperture, and  to the reflectivity of the target surface. 
To extract further insights, we have also derived valid approximate expressions of the HCRBs for relatively distant targets, considering the two different cases of known and known contours.

The obtained results showed that, if the contour of the target is known, the achievable accuracy in range and direction estimation is the same as in case of a point-like target given equal received energy. Thus, the variance of range estimation decreases linearly with the received energy and quadratically with the effective bandwidth; similarly, the variance of direction estimation decreases linearly with the received energy and quadratically with the array's effective aperture. 

A key outcome of the developed analysis is that knowledge of the contour has an enormous impact on the estimation accuracy of the range and orientation information. In Fig.~\ref{fig:HCRB} we observed that the difference in accuracy between known and unknown contour is up to three orders of magnitude. Nonetheless, the behavior of the direction accuracy is somewhat different compared to the range and orientation because as the vehicle moves farther away the accuracy for known/unknown contour converges. This is consistent with the intuition that, as the vehicle moves far away from the radar, it occupies a narrower range of angles and, consequently, the shape plays a less important role. Moreover, the parameter that is unique to the extended target --- the orientation --- turns out to depend on many parameters: received energy, bandwidth, array size, and even the reflection coefficient.

Finally, we have shown that lack of knowledge of the extended-target contour can be partially compensated by increasing the number of radars sensing the target vehicle from different angles as outlined in Fig.~\ref{fig:diversity}: for the considered parameters, four or more radars yield an average positioning error that is only twice as larger as for known contour. 

Some interesting directions of future work are now outlined. First, in this work we treated the surface roughness $\alpha$ as a known parameter for the sake of decoupling the effect of having some a priori knowledge of the vehicle contour and investigating its impact on the resulting HCRBs. As the two analytical results presented in Sec.~\ref{sub:known_contour} and Sec.~\ref{sec::HCRB_unknown} show, the lack of knowledge of the target contour already makes the derivation of the corresponding HCRBs and their analysis significantly more complex compared to the case of known contour. However, since $\alpha$ is generally unknown for the radar, an interesting future direction of research may be to extend the present analysis also to the case where the target reflectivity is unknown. In this sense, the HCRBs derived in this manuscript can be considered as more optimistic lower bounds.

Another important aspect concerns the correct choice of the level of detail (i.e., number of coefficients $Q$) to be used in the geometric contour model presented in Sec.~\ref{sec::contourmodel}. In this respect, methodologies similar to those used for the selection of an appropriate model order in traditional estimation problems (e.g., AOA estimation) can be considered, taking into account the existing trade-off between accuracy in the representation of the extended target and number of measurements required to get a reasonable contour estimate.

Lastly, a complete characterization of the theoretical localization accuracy in a 3D scenario can be of interest  outside the automotive context, e.g., in presence of aerial targets such as UAVs, and as such it deserves further investigation. Therefore, extending the proposed methodology and analyses to other application contexts represents a possible future research direction.

	\appendices

%\renewcommand{\thesection}{\Alph{section}}
%\renewcommand{\thesubsection}{\thesection.\arabic{subsection}}

	%\section{Fisher Information Matrix} %\label{app:FIM}
	
	\section{Proof of Theorem~
	\ref{prop1}}\label{app:FIS}
	%{Effective Fisher Information Matrix} 
	
	\subsection{General proof}\label{app1A}
	
	Define $\mathbf{d}_k \triangleq \mathbf{a}(\phi_k) (\sinTrunc (\phi_k-\beta_k))^{\alpha+1}\, s(t-2d_k/c)$, then \eqref{eq:received_signal} rewrites as $\mathbf{e} = g \sqrt{\ell_\text{R}} \sum_{k=1}^{K} h_k \mathbf{d}_k$.
	From \eqref{eq:log-likelihood_obs:derivatives} and using the definition of $\mathbf{S}_{22} = -\E \left[ \bigtriangleup_{\boldsymbol{\gamma}}^{\boldsymbol{\gamma}} \log p(\mathbf{y} |\boldsymbol{\theta}) \right]$,
	\begin{align}
	\mathbf{S}_{22} &=
	\frac{2g^2}{N_0}\ell_\text{R} \Re \int_{T} \sum_{k=1}^K \sum_{k'=1}^K \E \left[h_k^*h_{k'} \right]
	\frac{\partial\mathbf{d}_k^\mathsf{H}}{\partial \boldsymbol{\gamma}}
	\frac{\partial\mathbf{d}_{k'}}{\partial \boldsymbol{\gamma}^\mathsf{T}} \,\mathrm{d}t \label{eq:S22:start} \nonumber \\
	&= \frac{2g^2}{N_0}\ell_\text{R} 
	\sum_{k=1}^K \Re \int_{T}
	\frac{\partial\mathbf{d}_k^\mathsf{H}}{\partial \boldsymbol{\gamma}}
	\frac{\partial\mathbf{d}_k}{\partial \boldsymbol{\gamma}^\mathsf{T}} \,\mathrm{d}t
	\end{align}
	where we used the fact that $\E \left[h_k^* h_{k'}\right] =0$ for $k\neq k'$ and $\E \left[|h_k|^2\right] =1$.
	Calculation of $\partial\mathbf{d}_k^\mathsf{H}/\partial \boldsymbol{\gamma}$ is complicated because it depends on $\{u_k\}_{k=1}^K$, whose values in turn depend on the vehicle contour (i.e., $\{a_q, b_q\}_{q=1}^Q$). To proceed further, we approximate the sum by an integral similarly to what was done in \eqref{eq:aggregated_power:2}--\eqref{eq:aggregated_power:sum} where, on the contrary, the integral was approximated by a sum:
	\begin{align}
	\mathbf{S}_{22} &\approx \frac{2g^2}{N_0} \int_{\mathcal{C}} \Re \int_{T} \frac{\partial\mathbf{d}^\mathsf{H}(u)}{\partial \boldsymbol{\gamma}}
	\frac{\partial\mathbf{d}(u)}{\partial \boldsymbol{\gamma}^\mathsf{T}} \textbf{ }\,\mathrm{d}t\;\mathrm{d}\bm{r} \nonumber \\
	&= \frac{2g^2}{N_0} \int_{0}^{2\pi} \left(\Re \int_{T} \frac{\partial\mathbf{d}(u)^\mathsf{H}}{\partial \boldsymbol{\gamma}}
	\frac{\partial\mathbf{d}(u)}{\partial \boldsymbol{\gamma}^\mathsf{T}} \,\mathrm{d}t\right) \|\dot{\mathbf{r}}(u)\| \,\mathrm{d}u \label{eq:Fisher:S22}
	\end{align} 
	where for notation brevity 
	$$
	\mathbf{d}(u) = \mathbf{d}(\mathbf{r}(u)) =  \mathbf{a}(\phi) (\sinTrunc (\phi-\beta))^{\alpha+1}\, s(t-2d/c)
	$$
	and as usual we have omitted the dependency on $u$ of the vector of intermediate variables $\mathbf{\Theta} = [d\ \phi\ \beta]^\mathsf{T}$.
	By applying the chain rule $\frac{\partial\mathbf{d}^\mathsf{H}}{\partial\boldsymbol{\gamma}} = \frac{\partial\mathbf{\Theta}^\mathsf{T}}{\partial\boldsymbol{\gamma}} \frac{\partial\mathbf{d}^\mathsf{H}}{\partial\mathbf{\Theta}}$, $\mathbf{S}_{22}$ takes the form
	\begin{equation} \label{eq:Fisher:rest}
	\mathbf{S}_{22} = \frac{2g^2}{N_0} \int_{0}^{2\pi}
	\frac{\partial\mathbf{\Theta}^\mathsf{T}}{\partial\boldsymbol{\gamma}}
	\left(\Re \int_{T} \frac{\partial\mathbf{d}^\mathsf{H}}{\partial\mathbf{\Theta}} \frac{\partial\mathbf{d}}{\partial\mathbf{\Theta}^\mathsf{T}}
	\,\mathrm{d}t \right) \frac{\partial\mathbf{\Theta}}{\partial\boldsymbol{\gamma}^\mathsf{T}} \|\dot{\mathbf{r}}(u)\| \,\mathrm{d}u
	\end{equation} 
	where the dependency on $u$ was omitted for brevity.
	Using the fact that $\frac{\partial\mathbf{\Theta}^\mathsf{T}}{\partial\boldsymbol{\gamma}} = [\frac{\partial d}{\partial\boldsymbol{\gamma}}\ \frac{\partial \phi}{\partial\boldsymbol{\gamma}} \ \frac{\partial \beta}{\partial\boldsymbol{\gamma}}]$, whose entries are computed in Appendix~\ref{app:derivatives_intermediate_param}, and by expanding $\Re \int_{T} \frac{\partial\mathbf{d}^\mathsf{H}}{\partial\mathbf{\Theta}} \frac{\partial\mathbf{d}}{\partial\mathbf{\Theta}^\mathsf{T}}
	\,\mathrm{d}t$ using the formulas in  Appendix~\ref{app:2nd_derivatives}, we find that \eqref{eq:Fisher:rest} becomes
	\begin{multline}
	\mathbf{S}_{22} = \frac{2g^2N}{N_0}
	\int_{0}^{2\pi} w^2 \Bigg[
	\left(4\pi\frac{B_\text{RMS}}{c}\right)^2
	\boldsymbol{\mu}
	\boldsymbol{\mu}^\mathsf{T} \\
	+\frac{\pi^2}{12} (N^2-1) \cos^2(\phi)
	\boldsymbol{\eta}
	\boldsymbol{\eta}^\mathsf{T}
	+ \frac{1}{w^2}(\alpha+1)^2v^2 \boldsymbol{\xi} \boldsymbol{\xi}^\mathsf{T} \Bigg] \|\dot{\mathbf{r}}(u)\| \,\mathrm{d}u \label{eq:S22}
	\end{multline}
	where $w$ and $v$ are defined in \eqref{eq:w&v}, 
	$\boldsymbol{\mu} = \partial d/\partial\boldsymbol{\gamma}$, $\boldsymbol{\eta} = \partial \phi/\partial\boldsymbol{\gamma}$ and $\boldsymbol{\xi} = \boldsymbol{\eta} -\partial \beta /\partial\boldsymbol{\gamma}$. %Closed-form expressions for the latter derivatives are calculated in Appendix~\ref{app:derivatives_intermediate_param}.
	
	Similarly to the steps \eqref{eq:S22:start}--\eqref{eq:Fisher:rest}, the following expressions for $s_{11}$ and $\mathbf{s}_{21}$ in \eqref{eq:Fisher:submatrix} are also obtained:
	\begin{align}
	s_{11} 
	%	= -\E \bigtriangleup_{g}^{g} \tilde{p}
	&=\frac{2}{N_0}\ell_\text{R} \Re \int_{T} \E \left[ \left\|\sum_{k=1}^{K} h_k \mathbf{d}_k\right\|^2 \right] \mathrm{d}t \nonumber \\
	&\approx \frac{2}{N_0} \int_{0}^{2\pi} \left(\int_{T} \left\|\mathbf{d}\right\|^2 \mathrm{d}t\right) \|\dot{\mathbf{r}}(u)\| \,\mathrm{d}u \\
	%%%%%%%%%%%%%%%%%%%%%%%%%%%%%%%%%%%%%%%%%%%%%%%%%%%%%%%%%%%
	\mathbf{s}_{21} 
	%	= -\E \bigtriangleup_{\bar{\boldsymbol{\theta}}}^{g} \tilde{p} 
	&= \frac{2g}{N_0}\ell_\text{R} \Re \int_{T} \sum_{k=1}^K \sum_{k'=1}^K \E\left[h_k^* h_{k'} \right] \frac{\partial\mathbf{d}_k^\mathsf{H}}{\partial \boldsymbol{\gamma}}
	\mathbf{d}_{k'} \,\mathrm{d}t \nonumber \\
	&\approx \frac{2g}{N_0} \int_{0}^{2\pi}
	\frac{\partial\mathbf{\Theta}^\mathsf{T}}{\partial\boldsymbol{\gamma}}
	\left(\Re \int_{T} \frac{\partial\mathbf{d}^\mathsf{H}}{\partial\mathbf{\Theta}}
	\mathbf{d}
	\,\mathrm{d}t \right) \|\dot{\mathbf{r}}(u)\| \,\mathrm{d}u.
	\end{align}
	Combining the above expressions with the identities of Appendix~\ref{app:identities} and Appendix~\ref{app:other},
	\begin{align}
	s_{11} &= 
	%	\frac{2}{N_o}
	%	\int_{0}^{2\pi} \left\|\mathbf{a}(\phi)\right\|^2 \sinTrunc^{2(\alpha+1)}(\phi-\beta) \left(\int_{T}\left|s(t-\frac{2d}{c})\right|^2 \mathrm{d}t\right) \left\|\dot{\mathbf{r}}\right\| \,\mathrm{d}u \\
	\frac{2 N}{N_0} \int_{0}^{2\pi} w^2 \left\|\dot{\mathbf{r}}(u)\right\|\,\mathrm{d}u \label{eq:s11} \\
	%%%%%%%%%%%%%%%%%%%%%%%%%%%%%%%%%%%%%%%%%%%%%%%%%%%%%%%%%%%
	\mathbf{s}_{21} &= 
	\frac{2g}{N_0} \int_{0}^{2\pi}
	\frac{\partial\mathbf{\Theta}^\mathsf{T}}{\partial\boldsymbol{\gamma}}
	\left(\Re \int_{T} \frac{\partial\mathbf{d}^\mathsf{H}}{\partial\mathbf{\Theta}}
	\mathbf{d}
	\,\mathrm{d}t \right) \|\dot{\mathbf{r}}(u)\| \,\mathrm{d}u \nonumber \\
	&= \frac{2gN}{N_0} \int_{0}^{2\pi}
	w (\alpha+1)v\,\boldsymbol{\xi} \|\dot{\mathbf{r}}(u)\| \,\mathrm{d}u. \label{eq:S21}
	\end{align}
	Finally, by plugging \eqref{eq:S22}, \eqref{eq:s11} and \eqref{eq:S21} into the EFIM \eqref{eq:Fisher:submatrix}, and using the star product  \eqref{eq:star_product}, the final formula \eqref{eq:Fisher} follows.

	\subsection{Derivation of $\partial\mathbf{\Theta}^\mathsf{T}/\partial\boldsymbol{\gamma}$} \label{app:derivatives_intermediate_param}
	
	Explicit formulas for  $\partial\mathbf{\Theta}^\mathsf{T}/\partial\boldsymbol{\gamma} = [\partial d/\partial\boldsymbol{\gamma} \ \partial \phi/\partial\boldsymbol{\gamma} \ \partial \beta/\partial\boldsymbol{\gamma}]$ are listed next (we recall the functional dependencies $\boldsymbol{\rho}(\mathbf{m},\mathbf{n})$, $\mathbf{R}(\varphi)$, $\mathbf{p}(\mathring{d},\mathring{\phi})$, $\mathbf{r}(\mathring{d},\mathring{\phi},\varphi,\mathbf{m},\mathbf{n})$):
	\begin{align}
	\boldsymbol{\mu} \triangleq \frac{\partial d}{\partial\boldsymbol{\gamma}} =
	\begin{bmatrix}
	\partial d /\partial\mathring{d} \\ \partial d /\partial\mathring{\phi} \\ \partial d /\partial\varphi \\ \partial d /\partial \mathbf{m} \\ \partial d /\partial \mathbf{n}
	\end{bmatrix} &=
	\frac{1}{d}
	\begin{bmatrix}
	\mathring{d}+\mathring{d}^{-1}\boldsymbol{\rho}^\mathsf{T} \mathbf{R}^\mathsf{T} \mathbf{p}  \\
	\boldsymbol{\rho}^\mathsf{T} \mathbf{R}^\mathsf{T} \mathbf{p}_\perp \\
	-\boldsymbol{\rho}^\mathsf{T} \mathbf{R}^\mathsf{T} \mathbf{p}_\perp \\
	(\begin{bmatrix}1 & 0 \end{bmatrix}\mathbf{R}^\mathsf{T} \mathbf{r}) \boldsymbol{\sigma} \\
	(\begin{bmatrix}0 & 1 \end{bmatrix}\mathbf{R}^\mathsf{T} \mathbf{r}) \boldsymbol{\varsigma}
	\end{bmatrix}  \label{eq::mu} 
\end{align}
\begin{align}
	\boldsymbol{\eta} \triangleq \frac{\partial \phi}{\partial\boldsymbol{\gamma}} =
	\begin{bmatrix}
	\partial \phi /\partial\mathring{d} \\ \partial \phi /\partial\mathring{\phi} \\ \partial \phi /\partial\varphi \\ \partial \phi /\partial \mathbf{m} \\ \partial \phi /\partial \mathbf{n}
	\end{bmatrix} &= \frac{1}{d^2}
	\begin{bmatrix}
	\mathring{d}^{-1}\boldsymbol{\rho}^\mathsf{T} \mathbf{R}^\mathsf{T} \mathbf{p}_\perp \\
	\mathring{d}^{2} +\boldsymbol{\rho}^\mathsf{T} \mathbf{R}^\mathsf{T} \mathbf{p} \\
	\boldsymbol{\rho}^\mathsf{T} \mathbf{R}^\mathsf{T} \mathbf{r} \\
	(\begin{bmatrix}1 & 0 \end{bmatrix}\mathbf{R}^\mathsf{T} \mathbf{r}_\perp ) \boldsymbol{\sigma} \\
	(\begin{bmatrix}0 & 1 \end{bmatrix}\mathbf{R}^\mathsf{T} \mathbf{r}_\perp) \boldsymbol{\varsigma}
	\end{bmatrix} \label{eq::eta}
	\end{align}
\begin{align}
	\frac{\partial \beta}{\partial\boldsymbol{\gamma}} =
	\begin{bmatrix}
	\partial \beta /\partial\mathring{d} \\ \partial \beta /\partial\mathring{\phi} \\ \partial \beta /\partial\varphi \\ \partial \beta /\partial \mathbf{m} \\ \partial \beta /\partial \mathbf{n}
	\end{bmatrix} &=
	%	\frac{1}{\left\|\dot{\boldsymbol{\rho}}\right\|^2}
	%	\begin{bmatrix}
	%		0 \\ 0 \\ \left\|\dot{\boldsymbol{\rho}}\right\|^2 \\ 
	%		\dot{\boldsymbol{\gamma}} \begin{bmatrix}1 & 0 \end{bmatrix} \mathbf{R}^\mathrm{T} \dot{\mathbf{r}}_\perp \\
	%		\dot{\boldsymbol{\varsigma}} \begin{bmatrix}0&1 \end{bmatrix} \mathbf{R}^\mathrm{T} \dot{\mathbf{r}}_\perp
	%	\end{bmatrix} =
	\frac{1}{\left\|\dot{\boldsymbol{\rho}}\right\|^2}
	\begin{bmatrix}
	0 \\ 0 \\ \left\|\dot{\boldsymbol{\rho}}\right\|^2 \\ 
	- (\dot{\boldsymbol{\varsigma}}^\mathsf{T}\mathbf{n}) \dot{\boldsymbol{\sigma}} \\
	( \dot{\boldsymbol{\sigma}}^\mathsf{T}\mathbf{m}) \dot{\boldsymbol{\varsigma}}
	\end{bmatrix}
	\end{align}
	where $\mathbf{x}_\perp = \left(\begin{smallmatrix}0 & -1 \\ 1 & 0\end{smallmatrix}\right) \mathbf{x}$ for any arbitrary vector $\mathbf{x}$, $\dot{\boldsymbol{\sigma}}=\partial\boldsymbol{\sigma}/\partial u$, $\dot{\boldsymbol{\varsigma}}=\partial\boldsymbol{\varsigma}/\partial u$,  $\dot{\boldsymbol{\rho}}=\partial\boldsymbol{\rho}/\partial u$, and we recall from Appendix~\ref{app1A}  
	\begin{equation}
	    \boldsymbol{\xi} = \boldsymbol{\eta} -\partial \beta /\partial\boldsymbol{\gamma}. \label{eq::xi}
	\end{equation}

	\subsection{Derivation of $\Re \int_{T} \frac{\partial\mathbf{d}^\mathsf{H}}{\partial\mathbf{\Theta}} \frac{\partial\mathbf{d}}{\partial\mathbf{\Theta}^\mathsf{T}}
	\,\mathrm{d}t$} \label{app:2nd_derivatives}
	
	First, compute $\frac{\partial\mathbf{d}}{\partial\mathbf{\Theta}^\mathsf{T}}$:
	\begin{align}
	\frac{\partial\mathbf{d}}{\partial d} =&
	-\frac{2}{c} (\sinTrunc (\phi-\beta))^{\alpha+1}\mathbf{a}(\phi)
	\dot{s}\left(t-\frac{2d}{c}\right) \\
	%%%%%%%%%%%%%%%%%%%%%%%%%%%%%%%%%%%%%%%%%%%%%%%%%%%%%%%%%%%%%%%%%%%%%%%%%%%%%%%%%%%%%%%%%%%%%%
	\begin{split}
	\frac{\partial\mathbf{d}}{\partial \phi} =&
	(\sinTrunc (\phi-\beta))^{\alpha}
	\big[(\alpha+1) \cos(\phi-\beta) \mathbf{a}(\phi) \\
	&+\sinTrunc(\phi-\beta) \dot{\mathbf{a}}(\phi)\big]
	s\left(t-\frac{2d}{c}\right) 
	\end{split} \\
	%%%%%%%%%%%%%%%%%%%%%%%%%%%%%%%%%%%%%%%%%%%%%%%%%%%%%%%%%%%%%%%%%%%%%%%%%%%%%%%%%%%%%%%%%%%%%%
	\frac{\partial\mathbf{d}}{\partial \beta} =&
	-(\alpha+1) (\sinTrunc (\phi-\beta))^{\alpha} \cos(\phi-\beta) \mathbf{a}(\phi) s\left(t-\frac{2d}{c}\right)
	%%%%%%%%%%%%%%%%%%%%%%%%%%%%%%%%%%%%%%%%%%%%%%%%%%%%%%%%%%%%%%%%%%%%%%%%%%%%%%%%%%%%%%%%%%%%%%
	\end{align}
	where $\dot{s}(t) = \partial s(t)/\partial t$ and $\dot{\mathbf{a}}(\phi) = \partial \mathbf{a}(\phi)/\partial \phi$.
	In combination with the identities of Appendix~\ref{app:identities}, the expressions for all entries in $\Re\int_{T}\frac{\partial\mathbf{d}^\mathsf{H}}{\partial\mathbf{\Theta}}\frac{\partial\mathbf{d}}{\partial\mathbf{\Theta}^\mathsf{T}}\mathrm{d}t$ are
	\begin{align*}
	\Re	\int_{T}
	\frac{\partial\mathbf{d}^\mathsf{H}}{\partial d} \frac{\partial\mathbf{d}}{\partial d}
	\mathrm{d}t &=
	4 (\sinTrunc(\phi-\beta))^{2(\alpha+1)}  N \left(\frac{2\pi B_\text{RMS}}{c}\right)^2 \\
	%%%%%%%%%%%%%%%%%%%%%%%%%%%%%%%%%%%%%%%%%%%%%%%%%%%%%%%%%%%%%%%%%%%%%%%%%%%%%%%%%%%%%%%%%%%%%%%%%%%%%%%%%%%%%%%%%%%%%%%%%%%%%%%%%%%%%%%%%%
	\begin{split}
	\Re \int_{T} \frac{\partial\mathbf{d}^\mathsf{H}}{\partial \phi} \frac{\partial\mathbf{d}}{\partial \phi} \mathrm{d}t &=
	(\sinTrunc(\phi-\beta))^{2\alpha}  N \Big[(\alpha+1)^2 \cos^2(\phi-\beta) \\
	&+\sinTrunc^2(\phi-\beta)\cos^2(\phi) \frac{\pi^2}{12}(N^2-1)\Big] 
	\end{split}\\
	%%%%%%%%%%%%%%%%%%%%%%%%%%%%%%%%%%%%%%%%%%%%%%%%%%%%%%%%%%%%%%%%%%%%%%%%%%%%%%%%%%%%%%%%%%%%%%%%%%%%%%%%%%%%%%%%%%%%%%%%%%%%%%%%%%%%%%%%%%
	\Re \int_{T}
	\frac{\partial\mathbf{d}^\mathsf{H}}{\partial \beta} \frac{\partial\mathbf{d}}{\partial \beta}
	\mathrm{d}t &=  (\alpha+1)^2 (\sinTrunc(\phi-\beta))^{2\alpha} \cos^2(\phi-\beta) N \\
	%%%%%%%%%%%%%%%%%%%%%%%%%%%%%%%%%%%%%%%%%%%%%%%%%%%%%%%%%%%%%%%%%%%%%%%%%%%%%%%%%%%%%%%%%%%%%%%%%%%%%%%%%%%%%%%%%%%%%%%%%%%%%%%%%%%%%%%%%%
	\Re \int_{T}
	\frac{\partial\mathbf{d}^\mathsf{H}}{\partial \phi} \frac{\partial\mathbf{d}}{\partial \beta}
	\mathrm{d}t &= - (\alpha+1)^2 (\sinTrunc(\phi-\beta))^{2\alpha} \cos^2(\phi-\beta) N \\
	%%%%%%%%%%%%%%%%%%%%%%%%%%%%%%%%%%%%%%%%%%%%%%%%%%%%%%%%%%%%%%%%%%%%%%%%%%%%%%%%%%%%%%%%%%%%%%%%%%%%%%%%%%%%%%%%%%%%%%%%%%%%%%%%%%%%%%%%%%
	\Re	\int_{T}
	\frac{\partial\mathbf{d}^\mathsf{H}}{\partial d} \frac{\partial\mathbf{d}}{\partial \phi}
	\mathrm{d}t &= 
	\Re	\int_{T}
	\frac{\partial\mathbf{d}^\mathsf{H}}{\partial d} \frac{\partial\mathbf{d}}{\partial \beta}
	\mathrm{d}t = 0.
	\end{align*}
	The matrix form accepts a more succinct form,
	\begin{equation}
	\Re\int_{T}\frac{\partial\mathbf{d}^\mathsf{H}}{\partial\mathbf{\Theta}}\frac{\partial\mathbf{d}}{\partial\mathbf{\Theta}^\mathsf{T}}\mathrm{d}t =
	N (\sinTrunc(\phi-\beta))^{2\alpha}
	\left(\mathbf{g} \mathbf{g}^\mathsf{T} + \mathbf{G} \right)
	\end{equation}
	where
	\begin{align}
	\mathbf{g} &= (\alpha+1) \cos(\phi-\beta)
	\begin{bmatrix}
	0 & 1 & -1
	\end{bmatrix}^\mathsf{T} \\
	\mathbf{G} &=
	\sinTrunc^2(\phi-\beta)
	\begin{bmatrix}
	4 \left(\frac{2\pi B_\text{RMS}}{c}\right)^2 & 0 & 0 \\
	0 & \frac{\pi^2}{12}(N^2-1) \cos^2(\phi) & 0 \\
	0 & 0 & 0
	\end{bmatrix}.
	\end{align}

	\vspace{-0.4cm}
	%\section{Important Identities} \label{app:identities}
	\subsection{Useful identities} \label{app:identities}
	
	Consider a radar equipped with a ULA parallel to the ground such that its array response is $\mathbf{a}(\phi) = \exp(\mathrm{j} \pi (N-1)/2 \sin\phi) [1\ \exp(-\mathrm{j} \pi \sin\phi)\ \cdots\ \exp(-\mathrm{j} \pi (N-1)\sin\phi)]^\mathsf{T}$. Here, the reference element with phase 0 is taken at the center of the ULA because it results in the tightest bound. Then,
	\begin{align}
	\|\mathbf{a}(\phi)\|^2 &= N \\
	\Re \left(\mathbf{a}^\mathsf{H}(\phi) \dot{\mathbf{a}}(\phi)\right) &= \frac{1}{2}\frac{\partial \|\mathbf{a}(\phi)\|^2}{\partial\phi} =  0
	\end{align}
	and from $\dot{\mathbf{a}}(\phi) = \mathrm{j} \pi \cos(\phi) \diag \left(\frac{N-1}{2}, \ldots, -\frac{N-1}{2}\right) \mathbf{a}(\phi)$
	\begin{equation}
	\|\dot{\mathbf{a}}(\phi)\|^2 = \cos^2(\phi) \pi^2 (N-1)N(N+1)/12.
	%eliminated g^2(\phi)
	\end{equation}
	Regarding the signal waveform, by assumption $\int_{-\infty}^{+\infty} |s(t-\tau)|^2 \mathrm{d}t =1$, and using the Fourier transform we get
	\begin{multline}
	\int_{-\infty}^{+\infty} s^*(t-\tau) \dot{s}(t-\tau) \,\mathrm{d}t \\
	= \int_{-\infty}^{+\infty} \left[S(f) e^{-\mathrm{j} 2\pi\tau f}\right]^* \mathrm{j} 2\pi f S(f) e^{-\mathrm{j}2\pi\tau f} \mathrm{d}f \\
	= \mathrm{j} 2\pi \int_{-\infty}^{+\infty}  f\left|S(f)\right|^2 \mathrm{d}f = 0 \label{eq:center_freq}
	\end{multline}
	{\ed because by interpreting the integral $\int_{-\infty}^{+\infty}  f\left|S(f)\right|^2 \mathrm{d}f$ as the center of mass of the signal spectrum $\left|S(f)\right|^2$, the latter can be arbitrarily shifted in frequency in such a way that its center of mass is located at zero. In the usual case that the spectrum has even symmetry, the equality condition in \eqref{eq:center_freq} is readily satisfied. Moreover,}
	%is the baseband center frequency of the signal which can be shifted at will to 0, and
	\begin{multline*}
	\int_{-\infty}^{+\infty} \dot{s}^*(t-\tau) \dot{s}(t-\tau) \,\mathrm{d}t \\
	\!\!= \! \int_{-\infty}^{+\infty} \!\!\! \left[\mathrm{j} 2\pi f S(f) e^{-\mathrm{j} 2\pi\tau f}\right]^* \! \mathrm{j} 2\pi  f S(f) e^{-\mathrm{j} 2\pi\tau f} \mathrm{d}f 
	\!=\! (2\pi)^2 B_\text{RMS}^2
	\end{multline*}
	where $B_\text{RMS} \triangleq (\int_{-\infty}^{+\infty}  f^2 \left|S(f)\right|^2 \mathrm{d}f)^{1/2}$ is known as the RMS or effective bandwidth.

	\subsection{Derivation of $\Re \int_{T}\mathbf{d}^\mathsf{H} \frac{\partial\mathbf{d}}{\partial\mathbf{\Theta}^\mathsf{T}}\mathrm{d}t$} \label{app:other}
	
	With the help of the identities in Appendix~\ref{app:identities}, we find that
	\begin{align}
	\Re \int_{T}\mathbf{d}^\mathsf{H} \frac{\partial\mathbf{d}}{\partial d}\mathrm{d}t &= 0 \\
	\Re \int_{T}\mathbf{d}^\mathsf{H} \frac{\partial\mathbf{d}}{\partial \phi}\mathrm{d}t &=
	(\alpha+1) N (\sinTrunc(\phi-\beta))^{2\alpha+1} \cos(\phi-\beta) \\
	\Re \int_{T}\mathbf{d}^\mathsf{H} \frac{\partial\mathbf{d}}{\partial \beta}\mathrm{d}t &=
	- (\alpha+1) N (\sinTrunc(\phi-\beta))^{2\alpha+1} \cos(\phi-\beta).
	\end{align}
	Putting them together,
	\begin{equation}
	\Re \int_{T}\mathbf{d}^\mathsf{H} \frac{\partial\mathbf{d}}{\partial\mathbf{\Theta}^\mathsf{T}}\mathrm{d}t =
	N (\sinTrunc(\phi-\beta))^{2\alpha+1} \mathbf{g}^\mathsf{T}.
	\end{equation}

	\section{Proof of Theorem~\ref{lemma:large_distance}} \label{app:large_distance}

	Theorem~\ref{lemma:large_distance} equivalently is $\lim_{\mathring{d}\rightarrow\infty} \mathring{d}^{4} [\mathbf{J}(\boldsymbol{\gamma}) -(2 E/N_0)\mathbf{T}] = \mathbf{0}$, 
	or alternatively, $\lim_{\mathring{d}\rightarrow\infty} \mathring{d}^{4} \mathbf{J}(\boldsymbol{\gamma}) = \lim_{\mathring{d}\rightarrow\infty} \mathring{d}^{4} 2\frac{E}{N_0} \mathbf{T}$ if both limits exist.
	First, define the partition
	\begin{equation}
	\mathbf{J}(\boldsymbol{\gamma}) = 
	\begin{bmatrix}
	\mathbf{J}_{11} & \mathbf{J}_{12} \\
	\mathbf{J}_{12}^\mathsf{T} & \mathbf{J}_{22}
	\end{bmatrix}
	\end{equation}
	with the same block sizes than the partition of $\mathbf{T}$ in Theorem~\ref{lemma:large_distance}, and also partition the following vectors: $\boldsymbol{\mu} = [\boldsymbol{\mu}_1^\mathsf{T} \ \boldsymbol{\mu}_2^\mathsf{T}]^\mathsf{T}$, $\boldsymbol{\eta} = [\boldsymbol{\eta}_1^\mathsf{T} \ \boldsymbol{\eta}_2^\mathsf{T}]^\mathsf{T}$, $\boldsymbol{\xi} = [\boldsymbol{\xi}_1^\mathsf{T} \ \boldsymbol{\xi}_2^\mathsf{T}]^\mathsf{T}$, whose definition was given in   Appendix~\ref{app:derivatives_intermediate_param} and we recall that $\boldsymbol{\xi} = \boldsymbol{\eta} -\partial \beta /\partial\boldsymbol{\gamma}$. 
	The matrix equality is split into three matrix equalities:
	\begin{equation} \label{eq:matrix_inequalities}
	\lim_{\mathring{d}\rightarrow\infty} \mathring{d}^{4} \mathbf{J}_{nm} = \lim_{\mathring{d}\rightarrow\infty} \mathring{d}^{4} 2\frac{E}{N_0} \mathbf{T}_{nm}	
	\end{equation}
	where $(n,m) \in \{(1,1),(2,1),(2,2)\}$.
	For $(m,n)=(1,1)$ the left-hand side of \eqref{eq:matrix_inequalities} is
	\begin{multline} \label{eq:J11_limit}
	\lim_{\mathring{d}\rightarrow\infty} \mathring{d}^{4}\, \mathbf{J}_{11} 
	= \frac{2NG}{N_0}
	\int_{0}^{2\pi}
	\lim_{\mathring{d}\rightarrow\infty}\Big[
	w^2 
	\boldsymbol{\mu}_1
	\boldsymbol{\mu}_1^\mathsf{T} \\
	+w^2 M \cos^2(\phi)
	\boldsymbol{\eta}_1
	\boldsymbol{\eta}_1^\mathsf{T}
	+ (\alpha+1)^2 \operatorname{P}_w^\perp\left(v\,\boldsymbol{\xi}_1\right) \operatorname{P}_w^\perp\left(v\,\boldsymbol{\xi}_1^\mathsf{T}\right)  \Big]
	\|\dot{\mathbf{r}}\| \,\mathrm{d}u
	\end{multline}
	where the limit was passed inside the integral by  the monotone convergence theorem and we used the fact that $g^2 \mathring{d}^{4} =G$ is a constant. The limits are computed by applying them to each component that depends on $\mathring{d}$ separately, for instance, $\lim_{\mathring{d}\rightarrow\infty} w^2  \boldsymbol{\mu}_1 \boldsymbol{\mu}_1^\mathsf{T} = (\lim_{\mathring{d}\rightarrow\infty} w)^2  (\lim_{\mathring{d}\rightarrow\infty} \boldsymbol{\mu}_1) (\lim_{\mathring{d}\rightarrow\infty} \boldsymbol{\mu}_1)^\mathsf{T}$, resulting in
	\begin{align}
	& \!\!\int_{0}^{2\pi} \!\!	\lim_{\mathring{d}\rightarrow\infty} \!
	w^2  \boldsymbol{\mu}_1 \boldsymbol{\mu}_1^\mathsf{T} \|\dot{\mathbf{r}}\| \,\mathrm{d}u \! =\!
	\|\mathring{w}\|_\star^2
	\left[ \!\! \begin{array}{ccc}
	 L & \!\!\!\mathring{A} & \!\!\! -\mathring{A}\\
	\mathring{A} & \!\!\! \mathring{B}_1 & \!\!\! -\mathring{B}_1 \\
	-\mathring{A} & \!\!\! -\mathring{B}_1 & \!\!\! \mathring{B}_1
	\end{array} \!\!\right] \!\! \label{eq:term1} \\
	& \!\!\int_{0}^{2\pi} \!\!	\lim_{\mathring{d}\rightarrow\infty}
	w^2 M \cos^2(\phi) \boldsymbol{\eta}_1 \boldsymbol{\eta}_1^\mathsf{T} \|\dot{\mathbf{r}}\| \,\mathrm{d}u =
	\|\mathring{w}\|_\star^2
	\begin{bmatrix}
	0 & 0 & 0 \\
	0 & Z & 0 \\
	0 & 0 & 0
	\end{bmatrix}\\
	& \!\! \begin{multlined}
	\int_{0}^{2\pi}\!\! \lim_{\mathring{d}\rightarrow\infty} (\alpha+1)^2
	\operatorname{P}_w^\perp\left(v\,\boldsymbol{\xi}_1\right) \operatorname{P}_w^\perp\left(v\,\boldsymbol{\xi}_1^\mathsf{T}\right) \|\dot{\mathbf{r}}\| \,\mathrm{d}u \\
	= \|\mathring{w}\|_\star^2
	\begin{bmatrix}
	0 & 0 & 0 \\
	0 & \mathring{B}_2 & -\mathring{B}_2 \\
	0 & -\mathring{B}_2 & \mathring{B}_2
	\end{bmatrix} \label{eq:term3}
	\end{multlined}
	\end{align}
	where 
	\begin{align}
	\mathring{w} &= (\sinTrunc(\mathring{\phi}-\beta))^{\alpha+1} \\
	\mathring{A} &=  \|\mathring{w}\|_\star^{-2} \langle\mathring{w},\mathring{w}\,\bar{\mathbf{p}}_\perp^\mathsf{T}\mathbf{R}\boldsymbol{\rho} \rangle_\star \\
	\mathring{B}_1 &=  \|\mathring{w}\|_\star^{-2} \|\mathring{w}\,\bar{\mathbf{p}}_\perp^\mathsf{T}\mathbf{R}\boldsymbol{\rho}\|_\star^2  \\
	\mathring{B}_2 &= \|\mathring{w}\|_\star^{-2}(\alpha+1) \|\operatorname{P}_{\mathring{w}}^\perp(\mathring{v})\|_\star^2 \\
	\mathring{v} &= (\sinTrunc(\mathring{\phi}-\beta))^{\alpha} \cos(\mathring{\phi}-\beta)
	\end{align}
	Summing up \eqref{eq:term1}--\eqref{eq:term3} produces
	\begin{equation} \label{eq:J11_result}
	\lim_{\mathring{d}\rightarrow\infty} \mathring{d}^{4}\, \mathbf{J}_{11} =
	\frac{2N G \|\mathring{w}\|_\star^2}{N_0}
	\begin{bmatrix}
	L & \mathring{A} & -\mathring{A} \\
	\mathring{A} & \mathring{B}+\mathring{Z} & -\mathring{B} \\
	-\mathring{A} & -\mathring{B} & \mathring{B}
	\end{bmatrix}
	\end{equation}
	where $\mathring{B} = \mathring{B}_1+\mathring{B}_2$. Since $\lim_{\mathring{d}\rightarrow\infty}A = \mathring{A}$ and $\lim_{\mathring{d}\rightarrow\infty}B = \mathring{B}$,  the right-hand side of \eqref{eq:matrix_inequalities} for $(m,n)=(1,1)$ results in \eqref{eq:J11_result} too, thus proving the proof for $(m,n)=(1,1)$.
	
	For $(m,n) \in \{(2,1),(2,2)\}$, we follow the same procedure to prove that the left and right-hand side of \eqref{eq:matrix_inequalities} is
	\begin{align}
	\lim_{\mathring{d}\rightarrow\infty} \mathring{d}^{4} \mathbf{J}_{21} 
	%	= \lim_{\mathring{d}\rightarrow\infty} \mathring{d}^{4} 2\frac{E}{N_0} \mathbf{T}_{21} 
	&=	\frac{2NG \|\mathring{w}\|_\star^2}{N_0}
	\begin{bmatrix}
	\mathring{\mathbf{c}} & \mathring{\mathbf{q}} & -\mathring{\mathbf{q}}
	\end{bmatrix} \\
	\lim_{\mathring{d}\rightarrow\infty} \mathring{d}^{4} \mathbf{J}_{22} 
	%	= \lim_{\mathring{d}\rightarrow\infty} \mathring{d}^{4} 2\frac{E}{N_0} \mathbf{T}_{22} 
	&= \frac{2NG \|\mathring{w}\|_\star^2}{N_0}\mathring{\mathbf{T}}_{22},
	\end{align}
	where
	\begin{align}
	\mathring{\mathbf{c}} &= L \langle\mathring{w}\,\mathbf{s},\mathring{w} \rangle_\star/\|\mathring{w}\|_\star^2 \\ 
	\mathring{\mathbf{q}} &= L \frac{\langle\mathring{w}\,\mathbf{s},\mathring{w}\, \bar{\mathbf{p}}_\perp^\mathsf{T}\mathbf{R}\boldsymbol{\rho} \rangle_\star}{\|\mathring{w}\|_\star^{2} }
	+(\alpha+1)^2 \frac{\langle\mathring{\mathbf{t}},\operatorname{P}_{\mathring{w}}^\perp (\mathring{v}) \rangle_\star}{\|\mathring{w}\|_\star^{2} }\\
	\mathring{\mathbf{T}}_{22} &= \left(L \langle\mathring{w}\,\mathbf{s},\mathring{w}\,\mathbf{s} \rangle_\star +(\alpha+1)^2 \langle\mathring{\mathbf{t}},\mathring{\mathbf{t}} \rangle_\star\right) /\|\mathring{w}\|_\star^{2} \\
	\mathring{\mathbf{t}} &= \operatorname{P}_{\mathring{w}}^\perp \left(
	\mathring{v} \left\|\dot{\boldsymbol{\rho}}\right\|^{-2} \begin{bmatrix}
	(\dot{\boldsymbol{\varsigma}}^\mathsf{T}\mathbf{n})\dot{\boldsymbol{\sigma}}^\mathsf{T} &
	-(\dot{\boldsymbol{\sigma}}^\mathsf{T}\mathbf{m}) \dot{\boldsymbol{\varsigma}}^\mathsf{T}
	\end{bmatrix}^\mathsf{T}
	\right).
	\end{align}

%		\subsection{Derivation of $\mathbf{T}_{21}$ and $\mathbf{T}_{22}$} \label{app:expressions}
	
	The expressions for $\mathbf{T}_{21}$ and $\mathbf{T}_{22}$ then follow: %in Lemma~\ref{lemma:large_distance} are listed next:
	\begin{align}
	\mathbf{T}_{21} &= \begin{bmatrix}
		\mathbf{c} & \mathbf{q} & -\mathbf{q}
		\end{bmatrix} \label{eq::T21}\\
	\mathbf{T}_{22} &= L \langle\bar{w}\,\mathbf{s},\bar{w}\,\mathbf{s} \rangle_\star +(\alpha+1)^2 \langle\mathbf{t},\mathbf{t} \rangle_\star /\|w\|_\star^{2} \label{eq::T22}
		\end{align}
		with
	\begin{align}	
	\mathbf{c} &= L \langle\bar{w}\,\mathbf{s},\bar{w} \rangle_\star \label{eq::cbold}\\
	\mathbf{q} &= L \langle\bar{w}\,\mathbf{s},\bar{w}\, \bar{\mathbf{p}}_\perp^\mathsf{T}\mathbf{R}\boldsymbol{\rho} \rangle_\star 
	+(\alpha+1)^2 \langle\mathbf{t},\operatorname{P}_w^\perp (v) \rangle_\star /\|w\|_\star^{2} \label{eq::dbold}
	\end{align}
	and, in addition to the symbols defined in the statement of the theorem, 
	\begin{align}
	\mathbf{s} =&
	\begin{bmatrix}
	[1\ 0] \mathbf{R}^\mathsf{T} \bar{\mathbf{p}}\, \boldsymbol{\sigma}^\mathsf{T} &
	[0\ 1] \mathbf{R}^\mathsf{T} \bar{\mathbf{p}}\, \boldsymbol{\varsigma}^\mathsf{T}
	\end{bmatrix}^\mathsf{T} \label{eq:g} \\
	\mathbf{t} =& \operatorname{P}_w^\perp \left(
	v\left\|\dot{\boldsymbol{\rho}}\right\|^{-2} \begin{bmatrix}
	(\dot{\boldsymbol{\varsigma}}^\mathsf{T}\mathbf{n})\dot{\boldsymbol{\sigma}}^\mathsf{T} &
	-(\dot{\boldsymbol{\sigma}}^\mathsf{T}\mathbf{m}) \dot{\boldsymbol{\varsigma}}^\mathsf{T}
	\end{bmatrix}^\mathsf{T}
	\right). \label{eq:h}
	\end{align}

	%\vspace{-0.3cm}
	\section{Proof of Proposition~\ref{lemma:far_known}} \label{app:known_contour}

	The proof is articulated in two parts. First, we prove $\mathbf{T}_{11}$ is invertible, then we show that \eqref{eq:far_known} holds true.
	The determinant of $\mathbf{T}_{11}$ must satisfy $Z(BL-A^2)\neq0$, which decomposes into $Z\neq0$ and $A^2\neq BL$. Regarding the first condition, $Z>0$ unless the vehicle is at the ULA endfire  ($\mathring{\phi}=\pm\pi/2$). Regarding the second condition, observe that $A^2 = L^2 \langle\bar{w},\bar{w}\,\bar{\mathbf{p}}_\perp^\mathsf{T}\mathbf{R}\boldsymbol{\rho} \rangle_\star^2$ which is strictly smaller than $\|\bar{w}\,\bar{\mathbf{p}}_\perp^\mathsf{T}\mathbf{R}\boldsymbol{\rho}\|_\star^2$ unless $\mathbf{p}_\perp^\mathsf{T}\mathbf{R}\boldsymbol{\rho}(u)$ is constant for all $u$ in the illuminated part of the  contour. But $\mathbf{p}_\perp^\mathsf{T}\mathbf{R}\boldsymbol{\rho}(u)$ constant would imply a vehicle with no width, thus $A^2 < \|\bar{w}\,\bar{\mathbf{p}}_\perp^\mathsf{T}\mathbf{R}\boldsymbol{\rho}\|_\star^2 \leq BL$, concluding the first part of the proof.
	
	The EFIM in Theorem~\ref{lemma:large_distance} reduces to the $3\times3$ matrix $\mathbf{J}_{11} = 
	2(E/N_0) \mathbf{T}_{11} +o(\mathring{d}^{-4})\text{ as }\mathring{d}\rightarrow\infty$ because the vehicle contour is known, and proving~\eqref{eq:far_known} is equivalent to verifying $\lim_{\mathring{d}\rightarrow\infty} \mathring{d}^{-4} [\mathbf{C} -(2E/N_0)^{-1} \mathbf{T}_{11}^{-1}] = \mathbf{0}$ \cite[eq.~(2)]{balcazar1989nonuniform}, where we use the shorthand notation $\mathbf{C}=\mathbf{C}(\mathring{d},\mathring{\phi},\varphi)$. The left-hand side of the latter condition can be expressed as
	\begin{multline}
	\lim_{\mathring{d}\rightarrow\infty} \mathring{d}^{-4} \left[\mathbf{C} -\left(2\frac{E}{N_0}\right)^{-1} \mathbf{T}_{11}^{-1}\right]
	= \\ \lim_{\mathring{d}\rightarrow\infty} \left(\mathring{d}^{4} \mathbf{J}_{11}\right)^{-1} \left[\mathring{d}^{4}2\frac{E}{N_0} \mathbf{T}_{11} - \mathring{d}^{4}\mathbf{J}_{11} \right] \left(\mathring{d}^{4}2\frac{E}{N_0} \mathbf{T}_{11}\right)^{-1}
	\end{multline}
	because $\mathbf{C}=\mathbf{J}_{11}^{-1}$. The limit of a matrix product is the product of the limits if they are finite. Being the inverse a continuous function, $\lim_{\mathring{d}\rightarrow\infty} (\mathring{d}^{4} \mathbf{J}_{11})^{-1} = (\lim_{\mathring{d}\rightarrow\infty} \mathring{d}^{4} \mathbf{J}_{11})^{-1}$,  the latter computed in \eqref{eq:J11_result}. Same applies to $\lim_{\mathring{d}\rightarrow\infty}(\mathring{d}^{4}2E N_0^{-1} \mathbf{T}_{11})^{-1} = (\lim_{\mathring{d}\rightarrow\infty} \mathring{d}^{4}2E N_0^{-1} \mathbf{T}_{11})^{-1}$. All inverses can be proved to exist because $\mathbf{T}_{11}$ is invertible. Lastly, $\lim_{\mathring{d}\rightarrow\infty}[\mathring{d}^{4}2E N_0^{-1} \mathbf{T}_{11} - \mathring{d}^{4}\mathbf{J}_{11}] = 0$ by Theorem~\ref{lemma:large_distance}.

	%	\section{Approximation $L -A^2/B \approx L$} \label{app:star_product_approx}
	
	%	(see Appendix~\ref{} for more in-depth details {\color{red}I should include figure of the star-product for different contours and orientations}).
	%Let $\mathbf{R}\boldsymbol{\rho}(u_\times)$ be the point where the vehicle contour crosses the line connecting the radar and the vehicle spanned by $\mathbf{p}$, making $\mathbf{p}_\perp^\mathrm{T}\mathbf{R}\boldsymbol{\rho}(u_\times) = 0$, and assume the visible part of the vehicle contour begins at $u_0$ and ends at $u_1$, so that $u_0< u_\times<u_1$. By definition $w$ is a positive function for all $u\in[u_0,u_1]$ and 0 otherwise, whereas $\mathbf{p}_\perp^\mathrm{T}\mathbf{R}\boldsymbol{\rho}(u)$ is positive for $u\in[u_0,u_\times]$ and negative for $u\in[u_\times,u_1]$, and the same behavior applies to  $w\mathbf{p}_\perp^\mathrm{T}\mathbf{R}\boldsymbol{\rho}(u)$. Thus, in general $<\bar{w},\overline{w\,\mathbf{p}_\perp^\mathrm{T}\mathbf{R}\boldsymbol{\rho}}>$ will result in a low inner product, making $\mathbf{C}(\mathring{d}) \approx \mathbf{C}_\text{p.t.}(\mathring{d}) = (2E/N_0)^{-1}L^{-1}$.
	%{\color{red}Maybe include plot of $w\,\mathbf{p}_\perp^\mathrm{T}\mathbf{R}\boldsymbol{\rho}$ and $w\,\mathbf{p}_\perp^\mathrm{T}\mathbf{R}\boldsymbol{\rho}$ normalized and numerical evaluation of the star-product.}

	\vspace{-0.1cm}
	\section{Proof of Proposition~\ref{lemma:far_unknown}} \label{app:unknown_contour}
	
	Eq.~\eqref{eq:far_unknown} is equivalent to $[\mathbf{C}(\boldsymbol{\gamma})]_{1:3,1:3} = (2E N_0^{-1})^{-1} [\mathbf{T}^{-1}]_{1:3,1:3} +o(\mathring{d}^{4})$ because $[\mathbf{C}(\boldsymbol{\gamma})]_{1:3,1:3} = \mathbf{C}(\mathring{d},\mathring{\phi},\varphi)$ and it is easily verified that $[\mathbf{T}^{-1}]_{1:3,1:3} = \mathbf{U}$ by the block inversion formula.
	Therefore, it suffices to prove the more general statement $\mathbf{C}(\boldsymbol{\gamma}) = (2E N_0^{-1})^{-1} \mathbf{T}^{-1} +o(\mathring{d}^{4})$, which is equivalent to $\lim_{\mathring{d}\rightarrow\infty} \mathring{d}^{-4} [\mathbf{C} -(2E/N_0)^{-1} \mathbf{T}^{-1}] = \mathbf{0}$.
	The left-hand side can be expressed as
	\begin{multline}
	\lim_{\mathring{d}\rightarrow\infty} \mathring{d}^{-4} \left[\mathbf{C} -\left(2\frac{E}{N_0}\right)^{-1} \mathbf{T}^{-1}\right] 
	= \\ \lim_{\mathring{d}\rightarrow\infty} \left(\mathring{d}^{4} \mathbf{J}\right)^{-1}
	\left[\mathring{d}^{4}2\frac{E}{N_0} \mathbf{T} - \mathring{d}^{4}\mathbf{J} \right] \left(\mathring{d}^{4}2\frac{E}{N_0} \mathbf{T}\right)^{-1}
	\end{multline}
	since $\mathbf{C} = \mathbf{J}^{-1}$. As in the proof of Proposition~\ref{lemma:far_known}, we can prove that $\lim_{\mathring{d}\rightarrow\infty} \mathring{d}^{-4}(\mathring{d}^{4} \mathbf{J})^{-1}$ and $\lim_{\mathring{d}\rightarrow\infty} (\mathring{d}^{4}2E N_0^{-1}\mathbf{T})^{-1}$ exist, and that $\lim_{\mathring{d}\rightarrow\infty}[\mathring{d}^{4}2 E N_0^{-1} \mathbf{T} - \mathring{d}^{4}\mathbf{J}] = 0$ by Theorem~\ref{lemma:large_distance}, concluding the proof. The existence of the matrix inverses has been verified numerically.

	\bibliographystyle{IEEEtran}
	\bibliography{IEEEabrv,./references}

% Generated by IEEEtran.bst, version: 1.14 (2015/08/26)
\begin{thebibliography}{10}
\providecommand{\url}[1]{#1}
\csname url@samestyle\endcsname
\providecommand{\newblock}{\relax}
\providecommand{\bibinfo}[2]{#2}
\providecommand{\BIBentrySTDinterwordspacing}{\spaceskip=0pt\relax}
\providecommand{\BIBentryALTinterwordstretchfactor}{4}
\providecommand{\BIBentryALTinterwordspacing}{\spaceskip=\fontdimen2\font plus
\BIBentryALTinterwordstretchfactor\fontdimen3\font minus
  \fontdimen4\font\relax}
\providecommand{\BIBforeignlanguage}[2]{{%
\expandafter\ifx\csname l@#1\endcsname\relax
\typeout{** WARNING: IEEEtran.bst: No hyphenation pattern has been}%
\typeout{** loaded for the language `#1'. Using the pattern for}%
\typeout{** the default language instead.}%
\else
\language=\csname l@#1\endcsname
\fi
#2}}
\providecommand{\BIBdecl}{\relax}
\BIBdecl

\bibitem{patole2017automotive}
S.~M. Patole, M.~Torlak, D.~Wang, and M.~Ali, ``Automotive {R}adars: A {R}eview
  of {S}ignal {P}rocessing {T}echniques,'' \emph{IEEE Sign. Proc. Magaz.},
  vol.~34, no.~2, pp. 22--35, 2017.

\bibitem{ARSP_2021}
F.~Engels, P.~Heidenreich, M.~Wintermantel, L.~Stäcker, M.~Al~Kadi, and A.~M.
  Zoubir, ``Automotive {R}adar {S}ignal {P}rocessing: Research {D}irections and
  {P}ractical {C}hallenges,'' \emph{IEEE Journ. of Selec. Topics in Sign.
  Proc.}, vol.~15, no.~4, pp. 865--878, 2021.

\bibitem{SPM_2019}
I.~Bilik, O.~Longman, S.~Villeval, and J.~Tabrikian, ``The {R}ise of {R}adar
  for {A}utonomous {V}ehicles: Signal {P}rocessing {S}olutions and {F}uture
  {R}esearch {D}irections,'' \emph{IEEE Sign. Proc. Magaz.}, vol.~36, no.~5,
  pp. 20--31, 2019.

\bibitem{SPM_2019b}
G.~Hakobyan and B.~Yang, ``High-{P}erformance {A}utomotive {R}adar: A {R}eview
  of {S}ignal {P}rocessing {A}lgorithms and {M}odulation {S}chemes,''
  \emph{IEEE Sign. Proc. Magaz.}, vol.~36, no.~5, pp. 32--44, 2019.

\bibitem{dickmann2016automotive}
J.~Dickmann, J.~Klappstein, M.~Hahn, N.~Appenrodt, H.-L. Bloecher, K.~Werber,
  and A.~Sailer, ``Automotive {R}adar the {K}ey {T}echnology for {A}utonomous
  {D}riving: From {D}etection and {R}anging to {E}nvironmental
  {U}nderstanding,'' in \emph{IEEE Radar Conference}.\hskip 1em plus 0.5em
  minus 0.4em\relax IEEE, 2016, pp. 1--6.

\bibitem{ITS_2020}
Z.~Feng, M.~Li, M.~Stolz, M.~Kunert, and W.~Wiesbeck, ``Lane {D}etection with a
  {H}igh-{R}esolution {Au}tomotive {R}adar by {I}ntroducing a {N}ew {T}ype of
  {R}oad {M}arking,'' \emph{IEEE Trans. on Int. Transp. Sys.}, vol.~20, no.~7,
  pp. 2430--2447, 2019.

\bibitem{Poor_ADAS}
S.~Sun, A.~P. Petropulu, and H.~V. Poor, ``{MIMO} {R}adar for {A}dvanced
  {D}river-{A}ssistance {S}ystems and {A}utonomous {D}riving: {A}dvantages and
  {C}hallenges,'' \emph{IEEE Sign. Proc. Magaz.}, vol.~37, no.~4, pp. 98--117,
  2020.

\bibitem{frohle2018cooperative}
M.~Frohle, C.~Lindberg, and H.~Wymeersch, ``Cooperative {L}ocalization of
  {V}ehicles without {I}nter-{V}ehicle {M}easurements,'' in \emph{2018 IEEE
  Wireless Comm. and Netw. Conf. (WCNC)}.\hskip 1em plus 0.5em minus
  0.4em\relax IEEE, 2018, pp. 1--6.

\bibitem{soatti2018implicit}
G.~Soatti, M.~Nicoli, N.~Garcia, B.~Denis, R.~Raulefs, and H.~Wymeersch,
  ``Implicit {C}ooperative {P}ositioning in {V}ehicular {N}etworks,''
  \emph{IEEE Trans. on Int. Transp. Sys.}, vol.~19, no.~12, pp. 3964--3980,
  2018.

\bibitem{T_ITS}
A.~Fascista, G.~Ciccarese, A.~Coluccia, and G.~Ricci, ``{A}ngle of
  {A}rrival-{B}ased {C}ooperative {P}ositioning for {S}mart {V}ehicles,''
  \emph{IEEE Trans. on Int. Transp. Sys.}, vol.~19, no.~9, pp. 2880--2892,
  2018.

\bibitem{tsugawa2016review}
S.~Tsugawa, S.~Jeschke, and S.~E. Shladover, ``A {R}eview of {T}ruck
  {P}latooning {P}rojects for {E}nergy {S}avings,'' \emph{IEEE Trans. on Int.
  Vehic.}, vol.~1, no.~1, pp. 68--77, 2016.

\bibitem{mark2010principles}
\BIBentryALTinterwordspacing
M.~A. Richards, Ed., \emph{Principles of Modern Radar: Basic {P}rinciples},
  ser. Radar, Sonar \& Navigation.\hskip 1em plus 0.5em minus 0.4em\relax
  Institution of Engineering and Technology, 2010. [Online]. Available:
  \url{https://digital-library.theiet.org/content/books/ra/sbra021e}
\BIBentrySTDinterwordspacing

\bibitem{CFAR-FP}
A.~Coluccia, A.~Fascista, and G.~Ricci, ``{CFAR} {F}eature {P}lane: {A} {N}ovel
  {F}ramework for the {A}nalysis and {D}esign of {R}adar {D}etectors,''
  \emph{IEEE Trans. on Sign. Proc.}, vol.~68, pp. 3903--3916, 2020.

\bibitem{KNN}
A.~Coluccia, A.~Fascista, and G.~Ricci, ``A k-nearest neighbors approach to the
  design of radar detectors,'' \emph{Signal Processing}, vol. 174, p. 107609,
  2020.

\bibitem{KNN_2}
A.~Coluccia, A.~Fascista, and G.~Ricci, ``{A} {KNN}-{B}ased {R}adar {D}etector
  for {C}oherent {T}argets in {N}on-{G}aussian {N}oise,'' \emph{IEEE Sign.
  Proc. Letters}, vol.~28, pp. 778--782, 2021.

\bibitem{RangeSpread}
A.~Coluccia, A.~Fascista, and G.~Ricci, ``A novel approach to robust radar
  detection of range-spread targets,'' \emph{Signal Processing}, vol. 166, p.
  107223, 2020.

\bibitem{skolnik1981introduction}
M.~I. Skolnik, \emph{\BIBforeignlanguage{eng}{Introduction to {R}adar
  {S}ystems}}, 3rd~ed., ser. McGraw-Hill international editions. Electrical
  engineering series.\hskip 1em plus 0.5em minus 0.4em\relax Boston:
  McGraw-Hill, 2001.

\bibitem{brennan1961angular}
L.~Brennan, ``Angular {A}ccuracy of a {P}hased {A}rray {R}adar,'' \emph{IRE
  Transactions on antennas and propagation}, vol.~9, no.~3, pp. 268--275, 1961.

\bibitem{miller1978modified}
R.~Miller and C.~Chang, ``A {M}odified {C}ram{\'e}r-{R}ao {B}ound and its
  {A}pplications,'' \emph{IEEE Trans. on Inf. Theory}, vol.~24, no.~3, pp.
  398--400, 1978 \color{black}.

\bibitem{Shen1}
Y.~Han, Y.~Shen, X.-P. Zhang, M.~Z. Win, and H.~Meng, ``{P}erformance {L}imits
  and {G}eometric {P}roperties of {A}rray {L}ocalization,'' \emph{IEEE Trans.
  on Inf. Theory}, vol.~62, no.~2, pp. 1054--1075, 2016.

\bibitem{Shen2}
Y.~Wang, Y.~Wu, and Y.~Shen, ``{J}oint {S}patiotemporal {M}ultipath
  {M}itigation in {L}arge-{S}cale {A}rray {L}ocalization,'' \emph{IEEE Trans.
  on Sign. Proc.}, vol.~67, no.~3, pp. 783--797, 2019 \color{black}.

\bibitem{granstrom2016extended}
K.~Granstrom, M.~Baum, and S.~Reuter, ``Extended {O}bject {T}racking:
  Introduction, {O}verview and {A}pplications,'' \emph{arXiv:1604.00970}, 2016.

\bibitem{ICASSP2021}
G.~Yao, P.~Wang, K.~Berntorp, H.~Mansour, P.~Boufounos, and P.~V. Orlik,
  ``Extended {O}bject {T}racking {W}ith {A}utomotive {R}adar {U}sing
  {B}-{S}pline {C}hained {E}llipses {M}odel,'' in \emph{IEEE International
  Conference on Acoustics, Speech and Signal Processing (ICASSP)}, 2021, pp.
  8408--8412.

\bibitem{hammarstrand2012extended}
L.~Hammarstrand, L.~Svensson, F.~Sandblom, and J.~Sorstedt, ``Extended {O}bject
  {T}racking {U}sing a {R}adar {R}esolution {M}odel,'' \emph{IEEE Trans. on
  Aerosp. and Elect. Sys.}, vol.~48, no.~3, pp. 2371--2386, 2012.

\bibitem{zhang2005dynamic}
X.~Zhang, P.~Willett, and Y.~Bar-Shalom, ``Dynamic {C}ram{\'e}r-{R}ao {B}ound
  for {T}arget {T}racking in {C}lutter,'' \emph{IEEE Trans. on Aerosp. and
  Elect. Sys.}, vol.~41, no.~4, pp. 1154--1167, 2005 \color{black}.

\bibitem{Karl1}
K.~Granstrom and C.~Lundquist, ``On the {U}se of {M}ultiple {M}easurement
  {M}odels for {E}xtended {T}arget {T}racking,'' in \emph{Proceedings of the
  International Conference on Information Fusion}, 2013, pp. 1534--1541
  \color{black}.

\bibitem{buhren2006simulation}
M.~Buhren and B.~Yang, ``Simulation of {A}utomotive {R}adar {T}arget {L}ists
  {U}sing a {N}ovel {A}pproach of {O}bject {R}epresentation,'' in
  \emph{Intelligent Vehicles Symposium}.\hskip 1em plus 0.5em minus 0.4em\relax
  IEEE, 2006, pp. 314--319.

\bibitem{knill2016direct}
C.~Knill, A.~Scheel, and K.~Dietmayer, ``A {D}irect {S}cattering {M}odel for
  {T}racking {V}ehicles with {H}igh-{R}esolution {R}adars,'' in
  \emph{Intelligent Vehicles Symposium}.\hskip 1em plus 0.5em minus 0.4em\relax
  IEEE, 2016, pp. 298--303.

\bibitem{RectShape}
P.~Broßeit, M.~Rapp, N.~Appenrodt, and J.~Dickmann, ``Probabilistic
  {R}ectangular-{S}hape {E}stimation for {E}xtended {O}bject {T}racking,'' in
  \emph{IEEE Intelligent Vehicles Symposium (IV)}, 2016, pp. 279--285.

\bibitem{RectShape2}
X.~Cao, J.~Lan, X.~R. Li, and Y.~Liu, ``Extended {O}bject {T}racking {U}sing
  {A}utomotive {R}adar,'' in \emph{Intern. Conf. on Inf. Fusion (FUSION)},
  2018, pp. 1--5 \color{black}.

\bibitem{Karl2}
K.~Granstrom, S.~Reuter, D.~Meissner, and A.~Scheel, ``A {M}ultiple {M}odel phd
  {A}pproach to {T}racking of {C}ars {U}nder an {A}ssumed {R}ectangular
  {S}hape,'' in \emph{17th Intern. Conf. on Inf. Fusion (FUSION)}, 2014, pp.
  1--8.

\bibitem{Circle1}
N.~Petrov, A.~Gning, L.~Mihaylova, and D.~Angelova, ``{B}ox {P}article
  {F}iltering for {E}xtended {O}bject {T}racking,'' in \emph{15th Intern. Conf.
  on Inf. Fusion}, 2012, pp. 82--89.

\bibitem{Circle2}
D.~Angelova and L.~Mihaylova, ``{E}xtended {O}bject {T}racking {U}sing {M}onte
  {C}arlo {M}ethods,'' \emph{IEEE Trans. on Sign. Proc.}, vol.~56, no.~2, pp.
  825--832, 2008.

\bibitem{Ellipse1}
J.~W. Koch, ``{B}ayesian {A}pproach to {E}xtended {O}bject and {C}luster
  {T}racking {U}sing {R}andom {M}atrices,'' \emph{IEEE Trans. on Aerosp. and
  Elect. Sys.}, vol.~44, no.~3, pp. 1042--1059, 2008.

\bibitem{Ellipse2}
D.~Angelova, L.~Mihaylova, N.~Petrov, and A.~Gning, ``A {C}onvolution
  {P}article {F}iltering {A}pproach for {T}racking {E}lliptical {E}xtended
  {O}bjects,'' in \emph{16th Intern. Conf. on Inf. Fusion}, 2013, pp.
  1542--1549.

\bibitem{Parametric1}
X.~Cao, J.~Lan, and X.~R. Li, ``{E}xtension-{D}eformation {A}pproach to
  {E}xtended {O}bject {T}racking,'' \emph{IEEE Trans. on Aerosp. and Elect.
  Sys.}, vol.~57, no.~2, pp. 866--881, 2021 \color{black}.

\bibitem{GaussProc_Shape}
N.~Wahlström and E.~Özkan, ``Extended {T}arget {T}racking {U}sing {G}aussian
  {P}rocesses,'' \emph{IEEE Trans. on Sign. Proc.}, vol.~63, no.~16, pp.
  4165--4178, 2015.

\bibitem{GaussProc2_Shape}
X.~Tang, M.~Li, R.~Tharmarasa, and T.~Kirubarajan, ``Seamless {T}racking of
  {A}pparent {P}oint and {E}xtended {T}argets {U}sing {G}aussian {P}rocess
  {PMHT},'' \emph{IEEE Trans. on Sign. Proc.}, vol.~67, no.~18, pp. 4825--4838,
  2019 \color{black}.

\bibitem{GP1}
T.~Hirscher, A.~Scheel, S.~Reuter, and K.~Dietmayer, ``{M}ultiple {E}xtended
  {O}bject {T}racking {U}sing {G}aussian {P}rocesses,'' in \emph{19th Intern.
  Conf. on Inf. Fusion (FUSION)}, 2016, pp. 868--875.

\bibitem{GP2}
W.~Aftab, R.~Hostettler, A.~De~Freitas, M.~Arvaneh, and L.~Mihaylova,
  ``{S}patio-{T}emporal {G}aussian {P}rocess {M}odels for {E}xtended and
  {G}roup {O}bject {T}racking {W}ith {I}rregular {S}hapes,'' \emph{IEEE Trans.
  on Vehic. Techn.}, vol.~68, no.~3, pp. 2137--2151, 2019 \color{black}.

\bibitem{RHM_Shape}
M.~Baum and U.~D. Hanebeck, ``Extended {O}bject {T}racking with {R}andom
  {H}ypersurface {M}odels,'' \emph{IEEE Trans. on Aerosp. and Elect. Sys.},
  vol.~50, no.~1, pp. 149--159, 2014.

\bibitem{BSpline_Shape}
J.-L. Yang, P.~Li, and H.-W. Ge, ``{Extended {T}arget {S}hape {E}stimation by
  {F}itting {B}-{S}pline {C}urve},'' \emph{Journal of Applied Mathematics},
  vol. 2014, no. none, pp. 1 -- 9, 2014 \color{black}.

\bibitem{B-Spline1}
A.~Daniyan, S.~Lambotharan, A.~Deligiannis, Y.~Gong, and W.-H. Chen,
  ``{B}ayesian {M}ultiple {E}xtended {T}arget {T}racking {U}sing {L}abeled
  {R}andom {F}inite {S}ets and {S}plines,'' \emph{IEEE Trans. on Sign. Proc.},
  vol.~66, no.~22, pp. 6076--6091, 2018.

\bibitem{Karl3}
K.~Granstrom, P.~Willett, and Y.~Bar-Shalom, ``{A}n {E}xtended {T}arget
  {T}racking {M}odel with {M}ultiple {R}andom {M}atrices and {U}nified
  {K}inematics,'' in \emph{18th Intern. Conf. on Inf. Fusion (Fusion)}, 2015,
  pp. 1007--1014.

\bibitem{MultipleEllipses1}
J.~Lan and X.~R. Li, ``{T}racking of {M}aneuvering {N}on-{E}llipsoidal
  {E}xtended {O}bject or {T}arget {G}roup {U}sing {R}andom {M}atrix,''
  \emph{IEEE Trans. on Sign. Proc.}, vol.~62, no.~9, pp. 2450--2463, 2014
  \color{black}.

\bibitem{staib1989parametrically}
L.~H. Staib and J.~S. Duncan, ``Parametrically {D}eformable {C}ontour
  {M}odels,'' in \emph{Computer Society Conference on Computer Vision and
  Pattern Recognition}.\hskip 1em plus 0.5em minus 0.4em\relax IEEE, 1989, pp.
  98--103.

\bibitem{xu2009hybrid}
L.~Xu and X.~R. Li, ``Hybrid {C}ram{\'e}r-{R}ao {L}ower {B}ound on {T}racking
  {G}round {M}oving {E}xtended {T}arget,'' in \emph{Intern. Conf. on Inf.
  Fusion (FUSION)}.\hskip 1em plus 0.5em minus 0.4em\relax IEEE, 2009, pp.
  1037--1044.

\bibitem{tichavsky1998posterior}
P.~Tichavsky, C.~H. Muravchik, and A.~Nehorai, ``Posterior {C}ram{\'e}r-{R}ao
  {B}ounds for {D}iscrete-time {N}onlinear {F}iltering,'' \emph{IEEE Trans. on
  Sign. Proc.}, vol.~46, no.~5, pp. 1386--1396, 1998.

\bibitem{zhong2010comparison}
Z.~Zhong, H.~Meng, and X.~Wang, ``A {C}omparison of {P}osterior
  {C}ram{\'e}r--{R}ao {B}ounds for {P}oint and {E}xtended {T}arget
  {T}racking,'' \emph{IEEE Sign. Proc. Letters}, vol.~17, no.~10, pp. 819--822,
  2010.

\bibitem{PCRB_extend}
X.~Tang, M.~Li, R.~Tharmarasa, and T.~Kirubarajan, ``{P}osterior
  {C}ramér-{R}ao {L}ower {B}ounds for {E}xtended {T}arget {T}racking with
  {G}aussian {P}rocess {PMHT},'' in \emph{Intern. Conf. on Inf. Fusion
  (FUSION)}, 2019, pp. 1--8.

\bibitem{balcazar1989nonuniform}
J.~L. Balc{\'a}zar and J.~Gabarr{\'o}, ``Nonuniform {C}omplexity {C}lasses
  {S}pecified by {L}ower and {U}pper {B}ounds,'' \emph{RAIRO-Theoretical
  Informatics and Applications}, vol.~23, no.~2, pp. 177--194, 1989.

\bibitem{kulmer2018impact}
J.~Kulmer, F.~Wen, N.~Garcia, H.~Wymeersch, and K.~Witrisal, ``Impact of
  {R}ough {S}urface {S}cattering on {S}tochastic {M}ultipath {C}omponent
  {M}odels,'' in \emph{International Symposium on Personal, Indoor and Mobile
  Radio Communications (PIMRC)}.\hskip 1em plus 0.5em minus 0.4em\relax IEEE,
  2018, pp. 1410--1416.

\bibitem{std1979antennas}
\emph{Standard Test Procedures for Antennas}, IEEE Std. 149--1979, 1979.

\bibitem{jaeschke2014high}
T.~Jaeschke, C.~Bredendiek, S.~K{\"u}ppers, and N.~Pohl, ``High-{P}recision
  {D}-band {FMCW}-{R}adar {S}ensor based on a {W}ideband {SiGe}-{T}ransceiver
  {MMIC},'' \emph{IEEE Transactions on Microwave Theory and Techniques},
  vol.~62, no.~12, pp. 3582--3597, 2014.

\bibitem{van2004detection}
H.~L. Van~Trees, \emph{Detection, Estimation, and Modulation Theory - Part
  I}.\hskip 1em plus 0.5em minus 0.4em\relax John Wiley \& Sons, 2004.

\bibitem{noam2009notes}
Y.~Noam and H.~Messer, ``Notes on the {T}ightness of the {H}ybrid
  {C}ram{\'e}r--{R}ao {L}ower {B}ound,'' \emph{IEEE Trans. on Sign. Proc.},
  vol.~57, no.~6, pp. 2074--2084, 2009.

\bibitem{lees1989digital}
M.~Lees, ``Digital {B}eamforming {C}alibration for {FMCW} {R}adar,'' \emph{IEEE
  Trans. on Aerosp. and Elect. Sys.}, vol.~25, no.~2, pp. 281--284, 1989.

\end{thebibliography}
	
\end{document}